\documentclass[12pt]{article}
\usepackage{amsmath}
\usepackage{amsfonts} % no for mathdesign

\usepackage[T1]{fontenc}
\usepackage[mathletters]{ucs}
\usepackage[utf8x]{inputenc}

\usepackage{titlesec,comment}
\usepackage{natbib}
\usepackage[left=3.2cm,right=3.2cm]{geometry}
\usepackage[scale=0.9]{FiraMono}

\usepackage{mlmodern, XCharter}
\usepackage{bm}

\PassOptionsToPackage{unicode}{hyperref}
\PassOptionsToPackage{naturalnames}{hyperref}
\PassOptionsToPackage{hyphens}{url}
\PassOptionsToPackage{dvipsnames,svgnames*,x11names*}{xcolor}
\usepackage{xcolor}
\usepackage[hidelinks]{hyperref}
\hypersetup{
  pdftitle={The optimality of (stochastic) veto delegation},
  pdfauthor={Xiaoxiao HU; Haoran LEI},
  colorlinks=false,
  pdfcreator={LaTeX}
  }
\usepackage{longtable,booktabs,array}
\usepackage{calc} % for calculating minipage widths

\usepackage{etoolbox}
\makeatletter
\patchcmd\longtable{\par}{\if@noskipsec\mbox{}\fi\par}{}{}
\makeatother

\makeatletter
\def\fps@figure{!htbp}
\makeatother

\makeatletter
\usepackage{float}

\usepackage{amsthm}
\theoremstyle{plain}
\newtheorem{corollary}{Corollary}
\theoremstyle{definition}

\theoremstyle{plain}

\theoremstyle{plain}
\newtheorem{proposition}{Proposition}
\theoremstyle{remark}

\makeatother
\newlength{\cslhangindent}
\newlength{\csllabelwidth}
\newenvironment{CSLReferences}[2] % #1 hanging-ident, #2 entry spacing
 {% don't indent paragraphs
  \setlength{\parindent}{0pt}
  \ifodd #1 \everypar{\setlength{\hangindent}{\cslhangindent}}\ignorespaces\fi
  \ifnum #2 > 0
  \setlength{\parskip}{#2\baselineskip}
  \fi
 }%
 {}
\usepackage{calc}

\usepackage[shortlabels]{enumitem} 
\usepackage{subcaption}

\usepackage{tikz,xcolor}
\usetikzlibrary{patterns, fit, decorations.pathreplacing}
\usepackage{cleveref}
\def\set#1{\{ #1 \}}

\def\dif{\, d}

\def\pu{u_0}
\def\pv{v_0}

\def\inv{^{-1}}

\def\res{\hat{\theta}}
\def\E{\mathbb{E}}

\def\ta{\bar{a}}
\def\ttheta{\bar{\theta}}
\def\av{\hat{v}}
\def\pmap{\tilde{p}}
\def\amap{\tilde{a}}
\def\an{a_0}
\def\un{u_0}
\def\vn{v_0}

\newcommand{\state}{\theta}
\newcommand{\State}{\Theta}
\newcommand{\R}{\mathbb{R}}

\newcommand*{\ts}{\bar{\theta}}

\def\niu#1{{\color{black}#1}}

\newcommand*{\RR}{\mathbb{R}}

\begin{document}
\title{%
  The optimality of (stochastic) veto delegation%
  \thanks{\thx}
}
\long\def\thx{%
  We thank Navin Kartik and participants at
  various seminars and conferences for helpful comments. 
  All errors are ours.

  Xiaoxiao Hu, Ningbo University.
  Email: \url{xhuah@connect.ust.hk}.

  Haoran Lei, Hunan University.
  Email: \url{hleiaa@connect.ust.hk}.  
}
\author{Xiaoxiao Hu \and Haoran Lei}

\date{\today}

\maketitle

\begin{abstract}
We analyze the optimal delegation problem between a principal and an agent,
assuming that the latter has state-independent preferences.
We demonstrate that if the principal is more risk-averse than the agent 
toward non-status quo options, 
an optimal mechanism is a {\em veto mechanism}.
In a veto mechanism, the principal uses veto (i.e., maintaining the status quo)
to balance the agent's incentives and does not randomize among non-status quo options.
We characterize the optimal veto mechanism in a one-dimensional setting.
In the solution, the principal uses veto only when the state surpasses a critical threshold.

\bigskip

\textbf{Keywords:}
Optimal delegation, Veto delegation, Stochastic mechanism

\textbf{JEL Classification:} D72, D82
% A principal delegates decisions to an agent with state-independent preferences. If the principal is more risk-averse than the agent, an optimal mechanism is a {\em veto mechanism}.
% In this mechanism, the principal mixes the status quo decision with another
% state-dependent one. 
% We characterize the optimal mechanism 
% where states are one-dimensional, the agent is risk-neutral, 
% and the principal has quadratic preferences.
% In the solution, the principal uses veto (i.e., maintaining the status quo) only when the state surpasses a critical threshold.

\end{abstract}

\clearpage
% \begin{multicols}{2}
    {\small\tableofcontents}
% \end{multicols}

\clearpage  
% \linenumbers

\section{Introduction}\label{sec-intro}
In many settings of economic and political interest, a
principal consults an informed but biased agent while
contingent transfers between them are infeasible. The
principal then specifies a permissible set of options,
termed as the {\em delegation set}, and allows the agent to
choose any option from this set.
Holmstrom~(1980) analyzes the delegation problem when the
agent's utility depends on both the chosen option and the
realized state. 
He characterizes the optimal delegation sets,
focusing on the case when the delegated options form an interval (i.e., {\em interval delegation}).
There is a large literature following Holmstrom~(1980),%
    \footnote{We briefly review the literature on optimal 
    delegation in Section~\ref{sec-lit}.}
most of which assume state-dependent agent preferences and
emphasize the optimality of interval delegation.

Nevertheless, in many real-life interactions,
agents have \textit{state-independent} preferences and only care about the final decisions.
For example, prosecutors may seek their own favorable rulings regardless of the severity of defendants' crimes,
and financial advisors may recommend higher-commission products irrespective of their appropriateness.
When the agent has state-independent preferences,
any interval delegation set fails to elicit the agent’s private information and 
the principal cannot leverage the agent's expertise through interval delegation. 
On the other hand, interval delegation also seems incompatible with
the \emph{veto delegation} ubiquitous in many organizations,
where the principal can only either accept the agent's proposal or reject it in favor of an outside option. 

Veto delegation has been widely used in various organizations,
especially in the legislative processes and corporate governance. 
In the legislation, veto delegation is known as the \emph{closed rule},
where constituents can either approve or veto a committee's proposed bill.
The closed rule is deemed to constitute a ``critical component of managerial power
in the U.S. House of Representatives'' (Doran,~2010).
In corporate governance, boards of directors often can only disapprove proposals but cannot unilaterally enact new ones. 
More instances of the veto delegation in other organizations are documented in Marino~(2007),
Mylovanov~(2008) and Lubensky and Schmidbauer~(2018).

Motivated by these observations, we analyze the optimal delegation problem between 
a principal and an agent where the agent's preference is state-independent. 
In our model, the principal decides whether to maintain the status quo or choose another option.
To elicit information from the biased agent, the principal commits to a (possibly stochastic) direct mechanism.
We demonstrate that a principal's optimal mechanism must be a {\em veto mechanism}, 
provided that the principal is more risk-averse than the agent regarding non-status quo options. 
The veto mechanism is a specific direct mechanism in which
the principal never
randomizes among non-status quo options and only uses veto to balance the agent's incentives for truth-telling.
Our result rationalizes veto delegation as an arrangement
that achieves the optimal outcome.
% We further characterize the optimal veto mechanism 
% in a one-dimensional setting, where 
% the principal uses veto only when the state exceeds a critical threshold.

\subsection{Motivating example}\label{a-motivating-example}
To illustrate our main result, consider a political advisor (agent) advising a policy maker (principal).%
    \footnote{This example is adapted from the think-tank game in Lipnowski and Ravid~(2020), pp.~1632--1634.}
There are two equally likely states, $\theta_1$ and $\theta_2$,
and the agent privately observes the realized one.
The policy maker decides between a new policy ($a_1$ or $a_2$) or maintaining the status quo policy $a_0$.
The mild policy $a_1$ is suitable when the
state is $\theta_1$, and the aggressive policy $a_2$ is
suitable when the state is $\theta_2$.
The principal obtains a payoff of $1$
if the enacted new policy is suitable, or $-3$ otherwise;
the status quo $a_0$ yields a constant payoff of $0$ for the principal.
The advisor's preference ranks $a_2$ highest and $a_0$ lowest, regardless of the state. 
Suppose the advisor's payoff is $i$ when $a_i$ is chosen for $i \in \{0,1,2\}$.
Table~\ref{tab-payoff} shows players' payoffs for each policy--state combination.
% To illustrate our main result, consider a political advisor
% (agent) advising a policy maker (principal).%
% \footnote{%
%   This example is adapted from the
%   think-tank game in Lipnowski and Ravid~(2020), pp.~1632--1634.
% }
% There are two possible states, $\theta_1$ and $\theta_2$.
% It's commonly known that both states are equally likely and
% that the agent privately observes the realized state. The
% policy maker is contemplating whether to bring up a new
% policy from $\{a_1,a_2\}$ or to maintain the status quo
% policy $a_0$. The mild policy $a_1$ is suitable when the
% state is $\theta_1$, and the aggressive policy $a_2$ is
% suitable when the state is $\theta_2$. The policy maker
% wishes to maximize the social welfare: he obtains a payoff
% of $1$ if the implemented new policy is suitable or $-3$
% otherwise. Moreover, he obtains a constant payoff of $0$
% regardless of the state by maintaining the status quo policy
% $a_0$. On the other hand, the political advisor has
% political leanings: her most preferred policy is the
% aggressive policy $a_2$ and the least preferred policy is
% the status quo policy $a_0$ regardless of the state. Suppose
% that the advisor obtains payoff $i$ when policy $a_i$ is
% chosen for each $i \in \{ 0,1,2 \}$. Table~\ref{tab-payoff}
% depicts players' payoffs in each possible policy--state
% combination.

\begin{table}[H]
    \centering
    \caption{Payoffs of the policy maker and the political advisor}
    \label{tab-payoff}
 
    \begin{tabular}{@{}llll@{}}
    & $a_0$     & $a_1$        & $a_2$      \\ \cmidrule(l){2-4} 
    \multicolumn{1}{l|}{$\theta_1$} & \multicolumn{1}{l|}{$(0,0)$} & \multicolumn{1}{l|}{$(1,1)$}  & \multicolumn{1}{l|}{$(-3,2)$} \\ \cmidrule(l){2-4} 
    \multicolumn{1}{l|}{$\theta_2$} & \multicolumn{1}{l|}{$(0,0)$} & \multicolumn{1}{l|}{$(-3,1)$} & \multicolumn{1}{l|}{$(1,2)$}  \\ \cmidrule(l){2-4} 
    \end{tabular} 
\end{table}

The policy maker's ex ante optimal choice is $a_0$, yielding a payoff of $0$ for both players. 
Can the policy maker elicit information from the extremely biased advisor through delegation? 
The answer is yes. The policy maker can commit to the following 
mechanism where she may accept the advisor's proposal or {\em use veto} (i.e., choosing $a_0$): if the advisor proposes $a_1$, the policy maker will approve; if $a_2$ is proposed, 
the policy maker will approve or veto with equal probabilities. 
Given the principal's commitment,
the advisor is indifferent between proposing either new policy. 
Suppose the advisor proposes $a_i$ at each state $\theta_i$.
Then, the policy maker's expected utility is $3/4 > 0$.
The described veto mechanism can be viewed as the policy
maker delegating two \textit{lotteries} to the advisor: $l_1 = (1
\circ a_1)$ and $l_2 = (0.5 \circ a_0, 0.5 \circ a_2)$.%
\footnote{%
  Throughout this paper we denote a lottery with finite
  possible outcomes by $l = (p_1 \circ a_1, \dots, p_n \circ
  a_n)$ with $\sum_{i=1}^n p_i = 1$, meaning that outcome
  $a_i$ occurs with probability $p_i$ for each $i \in
  \set{1,\dots,n}$.
} 
Therefore, the policy maker benefits from delegating his
decisions to the agent when stochastic mechanisms are
allowed. 

One can verify that the described mechanism, while being
conceptually simple and relatively easy to implement, yields
the highest payoff for the principal among all
incentive-compatible mechanisms. The key feature of the veto
mechanism is that the principal only uses veto to balance
the agent's incentives; that is, she never randomizes
between non-status quo options. Later, we demonstrate the
optimality of the veto mechanism in a general environment.
In this sense, our paper contributes to understanding why veto delegation is
prevalent in various organizations.

\subsection{Literature}\label{sec-lit}
Our paper is mostly related to the literature on optimal delegation. Starting from Holmstrom~(1980), seminal works in this literature include Melumad and Shibano~(1991), Alonso and Matouschek~(2008), Kováč and Mylovanov~(2009), Amador and Bagwell~(2013), Kolotilin and Zapechelnyuk~(2019), etc.
Variants of the delegation model have been used in political economy (Bendor et al., 2001; Krishna and Morgan, 2001), monopoly regulation and organization (Baron and Myerson, 1982; Aghion and Tirole, 1997), tariff policies (Amador and Bagwell, 2013), etc. 
Despite the vast body of existing literature, this paper is, to the best of our knowledge, the first to explore the one-shot%
      \footnote{Frankel~(2006) and Chen~(2022) study the dynamic delegation problem with agent state-independent preferences.} 
optimal delegation problem where the agent has state-independent preferences. 
State-independent agent preferences imply that deterministic mechanisms are generally suboptimal.
Specifically, if the agent's utility function is injective,
no incentive-compatible deterministic mechanisms can elicit the agent's private information.
The potential non-optimality of deterministic mechanisms in delegation has been demonstrated in previous studies,
such as Kováč and Mylovanov~(2009, Section 4) and Kartik, Kleiner, and Van Weelden~(2021, Appendix~E).
In the same paper, Kartik, Kleiner, and Van Weelden also provide a practical real-life example of using stochastic delegation in nominating judges.

% The non-optimality of deterministic mechanisms in optimal delegation
% have been identified in the literature.%
% Unfortunately, while deterministic mechanisms can be interpreted (and implemented) as delegation,
% KKvW's interpretation: "We view such mechanisms as not only theoretically important, but also rele- vant in applications. For instance, consider a President (Proposer) nominating a judge for confirmation from the Senate (Vetoer) to a lifetime appointment. Both the President and the (pivotal voter in the) Senate have preferences over the ideology of the appointed judge. While some potential nominees have extensive written records or well-identified ideologies, others have less of a record of their own legal opinions and thus could be viewed as a lottery over ideologies."
% "For example, while John Roberts had argued several cases before the U.S. Supreme Court, he had served only two years as a judge prior to his nomination to the Court by President George W. Bush in 2005; in 2020, Vice President Mike Pence complained that “John Roberts has been a disappointment to conservatives”."
As we interpret our result as demonstrating the optimality of veto delegation, 
we discuss the relationship between this paper and the growing literature on veto delegation
(Dessein, 2002; Marino, 2007; Mylovanov, 2008; Lubensky and Schmidbauer, 2018).
The above-mentioned works differ from the optimal delegation literature as they focus on 
specific game forms of veto delegation---%
the agent proposes some option and the principal either accepts it or rejects it for an outside option.
Fixing an exogenous outside option, Dessein~(2002) compares full delegation
and veto delegation and shows that full delegation dominates veto delegation
as long as the conflict of interest is not extreme.
Lubensky and Schmidbauer~(2018) strengthen the result of Dessein (2002),
demonstrating that full delegation is superior by explicitly characterizing the
most informative veto equilibrium. 
Marino~(2007) challenges Dessein’s conclusion and shows the superiority 
of veto delegation in another setting with different
assumptions of players’ preferences and prior distribution. 
Mylovanov~(2008) allows the outside option to be chosen endogenously and 
shows the equivalence of optimal delegation and optimal veto delegation.
It is worth emphasizing that unlike those works mentioned above, 
we adopt a mechanism design approach and the
veto delegation arrangement emerges as the principal’s optimal mechanism.
Another important distinction is that in the above-mentioned works, the
principal decides whether to veto or not after the agent’s proposal; however, in
our paper the principal commits ex ante to his probability of veto conditional
on the agent's proposal to ensure information provision from the extremely
biased agent.

Finally, our paper relates to the broader information transmission literature. The assumption of sender (i.e., agent) state-independent preferences is common in the literature of communication with hard evidence
(Glazer and Rubinstein, 2004;
Glazer and Rubinstein, 2006; 
Hart et al., 2017) 
and information design (e.g., the judge-prosecutor game in Kamenica and Gentzkow (2011)).
There is also a growing literature on cheap talk with sender state-independent preferences (Chakraborty and Harbaugh, 2010; Lipnowski and Ravid, 2020; Diehl and Kuzmics, 2021). While this assumption is arguably of great empirical relevance, it has rarely been studied in the existing delegation literature. This paper therefore fills the gap.
\section{Model} \label{sec-model}
There are two players, a principal (he) and a better-informed agent (she).
The state~$θ$ follows some full-support distribution $μ ∈ Δ(Θ)$,
where the state space $Θ$ is a subset of $ℝ^n$ for some positive integer $n$.
The agent privately observes the realized state.
Denote by $a \in A$ a generic option or allocation,
where the set of available options $A$ is a subset of $ℝ^\ell$ for some positive integer $\ell$.
Principal's and agent's payoff functions are
$u \colon A × Θ → ℝ$ and $v \colon A → ℝ$, respectively.
While the principal's payoffs depend on the realized state, the agent's do not.
% Both $A$ and $Θ$ are compact sets that contain at least two elements. 
% Both players' utility functions are continuous.

\paragraph*{Status-quo option}
Notably, there exists a unique \textit{status quo option} $a_{0} \in A$, and 
any option distinct from $a_{0}$ is referred as a \textit{non-status quo option.}
We interpret choosing $\an$ as maintaining the status quo option or allocation. 
Real-life instances of $\an$ include customers buying nothing from the intermediaries,
investors rejecting the financial products recommended by the financial advisors,
and policy makers vetoing the political advisors' proposals for new policies. 
The principal is assumed to be \textit{more risk-averse} than the agent toward choosing non-status quo options at \textit{all} states.
Formally, for each $θ ∈ Θ$,
there exists a strictly concave transformation $h_θ: ℝ \to ℝ$  such that
$u(a, θ) = h_θ (v(a))$ for all $a \ne a_{0}$.
Additionally, the set $v (A \setminus \set{a_0}) = \set{ v(a) : a \in A \text{ and } a \ne \an }$
is assumed to be connected.
% The set $A \setminus \set{a_0}$ is assumed to be connected.

Following the literature on optimal delegation (e.g., Melumad and Shibano, 1991; Martimort and Semenov, 2006; Alonso and Matouschek, 2008),
we approach the delegation problem from a mechanism design perspective
and focus on the principal's optimal mechanism.
A \emph{direct mechanism} is a Borel measurable function $m : Θ \to Δ(A)$
such that expected payoffs are integrable.
To simplify notations, 
we extend the domains of players' payoff functions to incorporate stochastic options:  
$$
u(α, θ) := \int_A  u(a, θ)  \, d α
\quad \text{and} \quad
v(α) := \int_A  v(a)  \, d α
$$
where $α \in \Delta (A)$ denotes some stochastic option.
A mechanism $m$ is \textit{incentive compatible} if
\[
v(m(θ)) \ge v(m({θ'})) \text{ for all } θ, θ' ∈ Θ. \tag{IC}  
\]
Constraints (IC) can be reduced to a
set of indifference constraints:
\begin{equation}
  v(m(θ)) = \av \text{ for {\em all} $θ ∈ Θ$ and some $\av \in \RR$.}
  \label{eq-ic}
  \tag{$\mathcal{I}$}
\end{equation}
As the agent's preference is state-independent,
the reduced incentive-compatibility constraints \eqref{eq-ic} 
are also state-independent. That is,
in an incentive-compatible mechanism,
the agent obtains the same utility from any report regardless of the state.
Principal's maximization problem is given by:
\begin{equation}
\max_{m, \, \av ∈ ℝ} 
\int_\Theta u(m(θ), θ) \, d μ \text{ subject to constraints~\eqref{eq-ic}.}
\label{model:dm1}
\tag{$\mathcal{M}$}
\end{equation}

\begin{comment}
Denote by $m(θ) := α(θ) \in Δ(A)$ the stochastic option given the
reported state $\state$.
A mechanism $α$ is \textit{incentive-compatible} if
$$
\int_{A} v(\cdot) \dif m(θ) \ge \int_{A} v(\cdot) \dif α_{\res}$$ 
for all
$(\theta, \res) \in \Theta \times \Theta$.
Thanks to the functional form of the agent's utility, the incentive-compatible
constraints are reduced to a set of indifference constraints:
\begin{equation}
  \int_{A} v(\cdot) \dif α_\theta = \av \, \text{ for all $\theta \in \Theta$, where $\av$ is a constant.}
  \tag{$\mathcal I$} \label{eq-ic}
\end{equation}
The principal's optimization problem is:
\begin{equation}
  \quad \max_{\{ α_\theta \}_{\theta \in \Theta}, \av \in \R} 
  \int_\Theta  \left(\int_{A} u(a,\theta)  \dif α_\theta \right) \dif μ \quad 
\text{subject to constraints \eqref{eq-ic}.}
\tag{$\mathcal{M}$}
\label{model:dm1}
\end{equation}
\end{comment}
A direct mechanism $m^*$ is called an \textit{optimal mechanism} if 
$(m^*, \hat v^*)$ solves the problem~\eqref{model:dm1}
where $\hat v^* = v(m^* (θ))$ for all $θ$. 
% A direct mechanism $α^*$ is called an \textit{optimal mechanism} if these exists $v^* \in ℝ$ such that 
% $v^* \equiv v(α^*_θ)$ for all $θ$ and that $(α^*, v^*)$ solves the problem~\eqref{model:dm1}.
The solution concept is \textit{perfect Bayesian equilibrium} (henceforth, \emph{equilibrium}).
One interpretation of our mechanism design approach, as in most applied mechanism design papers, is to find an upper bound on the principal's welfare. Having said that, we also find the implementation via veto delegation fitting naturally into various contexts.

\subsection{Veto mechanisms}\label{sec-veto-def}
A direct mechanism $m$ is a {\em veto mechanism} if there exist two mappings,
$\pmap \colon Θ \to [0,1]$ and $\amap \colon Θ \to A \setminus \set{a_0}$,
such that for almost every state $θ$ with respect to the prior distribution $μ$, 
the induced stochastic option $m(θ)$ assigns probability $\pmap(\state)$ to the status quo option $\an$
and the complementary probability $1 - \pmap(\state)$ to the non-status quo option $\amap (\state)$.
This definition corresponds to the real-life veto delegation scenario where
the agent proposes some option $\amap(\theta)$
and then the principal vetoes that proposal 
with probability $\pmap(\state)$.
We refer to $\pmap (θ)$ as the \textit{veto probability}.
The veto mechanism encompasses all deterministic mechanisms as well as a class of simple stochastic mechanisms.
In these stochastic mechanisms,
whenever the principal makes a random choice,
the randomization occurs only between the status quo and one other option.

\Cref{prp:main} allows us to focus on veto mechanisms when
searching for an optimal mechanism.

% , making the principal's maximization problem 
% solvable using the standard Lagrangian method.

\begin{proposition}\label{prp:main}
  A direct mechanism is optimal {\em only if} it is a veto mechanism.
\end{proposition}
\begin{proof}
  Principal's maximization problem~\eqref{model:dm1} can be solved sequentially.
  First, fix some agent utility $\hat v$ and solve the following optimization problem
  \begin{equation}
  \pi(\hat v) := 
  \max_{m} 
  \int_\Theta u(m(θ), θ) \, d μ \text{ subject to } v(m(θ)) = \hat v \text{ for all } θ.
  \tag{$\mathcal{M}_{\hat v}$}
  \label{op:fixed-vhat}
  \end{equation}    
  Then, solve for $v^* \in \arg\max_{\hat v ∈ \RR} \pi(\hat v)$ that maximizes principal's ex ante payoff.
  For our purpose here, it suffices to show that for all $\hat v \in v(A)$, 
  any direct mechanism $\hat{m}$ that solves the problem \eqref{op:fixed-vhat} must be a veto mechanism.
  
  Fixing some $\hat v \in v(A)$, mechanism $\hat{m}$ is a solution to \eqref{op:fixed-vhat}
  only if, for $μ$-almost all $θ$, $\hat{m} (\state)$ solves the following maximization problem: 
  %%%%%%%%
  \begin{equation}
    \max_{\alpha\in Δ (A)} u(\alpha, θ)
      \text{ subject to $v(\alpha) = \hat{v}$.%  
        \footnotemark} 
      \label{eq:ss-condition} 
      \tag{$\mathcal{M}_{\hat{v}, θ}$}
  \end{equation}
  %%%%%%%%%  %%%%%%%
  Fixing some state $θ$, the objective function in \eqref{eq:ss-condition} can be written as
  $$
  u(α, θ) = \int_{A \setminus \set{a_0}}  u(a, θ)  \, d α + p_0 u (a_0, θ)
  $$ 
  where $p_0$ is the probability assigned to $a_0$ under $α$.
  Further,
  $$
  \begin{aligned}
  \int_{A \setminus \set{a_0}}  u(a, θ)  \, d α 
    & = (1 - p_0) \int_{A \setminus \set{a_0}}  h_θ \big( v(a) \big) \, d \big(\frac{α}{1 - p_0} \big)   \\
    & \le (1 - p_0) h_θ \Big( \int_{A \setminus \set{a_0}} v(a) \, d \big(\frac{α}{1 - p_0} \big) \Big)  \quad \text{(because $h_θ$ is concave)} \\
    & = (1 - p_0) h_θ ( v(a')) \text{ for some $a' \in A \setminus \set{a_0}$}  \quad \text{(because $v(A \setminus \set{a_{0}})$ is connected)} \\
    & = (1 - p_0) u ( a', θ).
  \end{aligned}
  $$
  Since function $h_θ$ is strictly concave, the inequality takes equality if and only if the support of 
  $α$ contains at most one option distinct from $a_0$.
  Therefore, any solution to the problem \eqref{op:fixed-vhat} must be a veto mechanism.
  %%%%%%%%%%
  %% footnote below
  \footnotetext{
  Precisely, $\hat{m}$ is a solution to maximization problem~\eqref{op:fixed-vhat} if and only if
  $v(\hat m(θ)) = \hat v$ for all $θ \in {Θ}$ and
  there exists some subset $\hat{Θ}$ of $Θ$ with $μ (\hat{Θ}) = 1$ such that $\hat{m}$ solves  
  \eqref{eq:ss-condition} for all $θ \in \hat{Θ}$.
  }
\end{proof}

The optimality of veto mechanisms relies on
the assumption that the principal is more risk-averse than the agent towards 
non-status quo options. 
If the principal is less risk-averse than the agent, the optimality of veto mechanisms may fail. Below, we provide an example to illustrate this scenario.
Suppose $Θ = [0,1] \subseteq ℝ$ and $A = (-\infty, 1] \cup \set{a_0} \subseteq ℝ$,
where the status quo option is denoted by some real number $\an > 1$. 
The principal has absolute loss payoff: $u(a,\theta) = - |a-\theta|$ for $a \in ℝ$ and
$u(a_0, \theta) = u_0 < 0$ for all $\theta$.
The agent has quadratic loss payoff $v(a) = -(1-a)^2$ for $a \in (-\infty,1]$ and $v(a_0) = v_0 \in ℝ$.
Then, the principal can achieve an outcome arbitrarily close 
to the first-best outcome without using veto at all.%
\footnote{%
    This almost first-best outcome is essentially the same with the example of non-optimality of deterministic allocations in Kováč and  Mylovanov~(2009, Section~4).
    A related example can be found in Example~E.1 in Kartik et al.~(2021).}
For arbitrarily small $\varepsilon > 0$, given the reported state $\theta$
the principal chooses a lottery of actions with
expectation being $\theta - \varepsilon$,
its support lying in $(-\infty, \theta]$ and 
the appropriate variance such that constraints~\eqref{eq-ic} hold.

\subsection{Remarks on the status quo option}

The proof of \Cref{prp:main} can be extended to show that
if the principal is more risk-averse than the agent towards 
\textit{all} options at all states (and $v(A)$ is connected),
an optimal mechanism must be deterministic.
In this case, the principal generally cannot elicit any information from the agent due to the state-independent agent preference.
Therefore, for the delegation problem to be non-trivial,
despite the fact that
the principal is more risk-averse than the 
agent towards all non-status quo options at all states,
adding the status quo option will make the principal no longer more risk-averse than the agent.
In other words, some player must have different 
risk attitudes between the status quo option and non-status quo options.
%the principal cannot be more risk-averse than the 
%agent towards all options at all states, 

In behavioral economics,
it is well-documented that 
individuals perceive qualitative differences between
choosing the status quo option and a non-status quo option
(e.g., Kahneman and Tversky, 1979).
Additionally, individuals often feel more responsible when enacting a non-status quo option (Bartling and Fischbacher, 2012).
In our motivating example, 
the policy maker is insensitive to the outcome of maintaining the status quo in that choosing $\an$ yields the same payoff at both states.
One possible reason is that
the policy maker feels more responsible for the 
outcome of a non-status quo policy
than that of maintaining the status quo.

% In many real-life situations, 
% the status quo option $a_0$ yields constant payoffs for the principal
% regardless of the state.
% For an illustration, consider that a customer ``delegates''
% her purchasing decision to a
% salesperson, and then the salesperson
% recommends some product to the customer.
% The state of nature measures the degree to which 
% the product matches the customer's needs, 
% and the customer is uncertain about the state.
% As the value of the alternative use of money for 
% the customer does not depend on how well the recommended product suits her,
% the payoff from rejecting the salesperson's recommendation
% is state-independent.
% %Similar arguments hold for the
% %investor-entrepreneur and manager-worker relationships. 
% In the literature of contract theory, it is also usually assumed that
% the value of outside option from not signing the contract
% is state-independent for the principal.
% For example, in an investor-entrepreneur relationship,
% the participation constraint requires that the 
% expected return to the investor (principal) must cover the 
% opportunity cost of funds, and the opportunity cost 
% is usually measured by a constant interest 
% rate (e.g., Gale and Hellwig, 2003).

% In our motivating example, the policy maker may 
% be (or feel) more responsible for the 
% actual effects of a new policy, and may be (or feel) less responsible for maintaining the status quo option. 
% Therefore, the policy maker can be insensitive to the result 
% of an inaction, and maintaining the status quo yields a constant payoff for him regardless of the state.
\section{Characterization}
We characterize the optimal mechanism in a 
one-dimensional setting.
Assume $Θ = [0,1] \subseteq \RR$ and $A = \set{\an} \cup [-M, M] \subset \RR$, where
$M$($>1$) is a sufficiently large positive number
and the status quo option is denoted by some real number
$\an < -M$. The state $\theta$ follows some full-support distribution $\mu \in \Delta(\Theta)$.
Agent's payoff function is
$v(a) = a \text{ for } a \in [-M,M]$
and $v(\an) = \vn \in \RR$.
Principal's payoff function is
$u(a, \theta) = - (a - \theta)^2 \text{ for } a \in [-M,M]$
and $u(\an, θ) = \un$ for all $θ \in Θ$.
While the principal's payoffs from a non-status quo option depend on the state, those from the status quo option do not.%   
\footnote{
  The assumption that principal's payoffs from $\an$ are state-independent  
  aligns with many real-life scenarios.
  For instance, in a customer-salesperson interaction,
  the customer ``delegates''
  his purchasing decision to the
  salesperson, and
  the state captures the degrees to which 
  the available products match the customer's needs. 
  As the value of the alternative use of money for 
  the customer does not depend on how well the products suit her,
  the payoffs from rejecting the salesperson's recommendation
  are state-independent.
  In the literature of contract theory, it is also usually assumed that
  the value of outside option from not signing the contract
  is state-independent for the principal.
  For example, in an investor-entrepreneur relationship,
  the participation constraint requires that the 
  expected return to the investor (principal) must cover the 
  opportunity cost of funds, and the opportunity cost 
  is usually measured by a constant interest rate (e.g., Gale and Hellwig, 2003).}

We impose the following restrictions. Assume $\un \le 0$.
Otherwise, the status quo option $\an$ always yields the unique highest payoff for the principal among all options, and the principal trivially chooses $\an$ regardless of the state. 
We also assume $\vn \le 0$ when characterizing the optimal mechanism and discussing comparative statics.
Later, we show that allowing $\pv$ being positive can drastically affect the
optimal mechanism (\Cref{sec-exten}).

\subsection{Analysis}\label{sec-characterizing}
We first briefly discuss the extreme case when $\pu=0$.
In this case, the principal can never do better than opting for 
the status quo option, and there exists a \textit{pooling equilibrium} in which
he maintains the status quo regardless of the state.
Nevertheless, there also exists a \textit{fully separating equilibrium},
in which the agent is strictly better off. 
In the separating equilibrium, the principal commits to the veto mechanism
$(\amap^*, \pmap^* )$ where $\amap^*(\theta) = \theta$ and $\pmap^*(\theta) = {\theta}/{(\theta - \pv)}$.
That is, given the agent's reported state $θ$,
the principal uses veto with probability $\frac{θ}{θ - \pv}$
and chooses option $\amap( θ ) = θ$ with the complementary probability.
This mechanism is incentive-compatible, and the principal's ex ante payoff 
is the same as that of the pooling equilibrium.
The separating equilibrium is the limit of
equilibria with $\pu< 0$, as discussed later.

From now on, assume $\pu< 0$.
By \Cref{prp:main}, we focus on veto mechanisms and 
the principal's maximization problem \eqref{model:dm1} reduces to
\begin{equation}\label{model:dm_M'}\tag{$\mathcal{M'}$}
  \begin{split}
  &\max_{\{ \pmap(\theta), \amap(\theta) \}_{\theta \in \Theta}, \hat{v}} 
  \int_\Theta  \big(\pmap(\theta) \pu- (1 - \pmap(\theta) )  (\amap(\theta) -\theta )^2 \big)  \, d μ \\
      \text{ subject to } 
  &\text{ (i) } 
    \forall\theta \in \Theta, \, \pmap(\theta) \pv+ (1 - \pmap(\theta)) \amap(\theta) = \hat{v};\\
  &\text{ (ii) } 
    \forall \theta \in \Theta, \pmap(\theta) \in [0,1] \text{ and } \amap(\theta) \in [-M, M].  
  \end{split}  
\end{equation}
Perturbing $\pmap({\state'})$ and $\amap({\state'})$
at some state $\state'$ will not impact the validity of 
constraints~\eqref{eq-ic} for other states $\theta \neq \theta'$,
as long as the agent's payoff remains fixed at $\hat{v}$ when reporting $\theta'$. 
Due to this observation,
we first fix some agent utility $\hat{v}$ and solve the principal's maximization problem 
\textit{state-wisely}. Then, we solve for the optimal $\hat{v}$.

The original problem~\eqref{model:dm_M'} is reduced to the following two-step maximization problem:
\begin{enumerate}
\item
  For each $θ \in \Theta$ and each agent's utility $\hat{v}\in [\vn,1]$, 
  solve the following optimization problem:
  \begin{equation}\label{model:dm3}\tag{$\mathcal{M}_1$}
  \pi (\hat{v}, \theta) \equiv 
    \max_{\pmap(\theta), \amap(\theta)}
    \pmap(\theta) \pu- (1 - \pmap(\theta))  (\amap(\theta) -\theta)^2 
  \end{equation} \[
  \begin{split}
  \text{ 
  subject to } \text{(i) } &
  \pmap(\theta) \pv+ (1 - \pmap(\theta)) \amap(\theta) = \hat{v};\\
  \text{(ii) } & \pmap(\theta) \in [0,1] \text{ and } \amap(\theta) \in [-M,M].
  \end{split}
  \]
\item
  Solve for the optimal $\hat{v}$:
  \begin{equation}
    \max_{\hat{v}\in [\pv, 1] } \int_\Theta \pi (\hat{v}, \theta) \, d \mu. \tag{$\mathcal{M}_2$} \label{model:dm4}
  \end{equation}
\end{enumerate}
We solve the sequential problems \eqref{model:dm3} and \eqref{model:dm4} using the standard Lagrangian method.
The optimal veto mechanism is summarized in \Cref{prp-char}.

\begin{proposition}\label{prp-char}
Assume $\pv \le 0$.
\begin{enumerate}[(a)]
\item \label{itm:non-trivial}
  If $(1 - \pv)^2 + \pu> (\mathbb{E}_μ (\theta) - \pv)^2$, then define
  {$\eta (θ) \equiv \sqrt{ (θ - \pv) ^ 2 + \pu} + \pv$} for
  $θ \ge \sqrt{- \pu} + \pv$ and the optimal veto mechanism is 
\begin{equation}
\begin{aligned}
    \amap^*(\theta) = 
          \begin{cases}
          {\bar a},  & \text{ if } \theta < { \bar \theta} ;\\
          \eta(\theta) , & \text{ if } \theta \ge \bar{\theta},
          \end{cases}
    \qquad %
      \pmap^*(\theta) = 
        \begin{cases}
          0,  & \text{ if } \theta <{\bar{\theta}} ; \\
          \frac{\eta(\theta) -  {\bar a}}{\eta(\theta) - \pv} , & \text{ if } \theta \ge {\bar{\theta}},
        \end{cases}
    \label{eq:optim-veto} 
\end{aligned}
\end{equation}
where $\bar \theta \equiv \eta^{-1} (\bar a)$ and $ {\bar a}$ is determined by
{$\int_{\eta \inv (\bar a)}^{1} (\eta (\theta) - \bar a) \dif μ(\theta) + \bar a - \mathbb{E}_μ (\theta) = 0$}.
Moreover, $\bar a \in (0, \mathbb{E}_μ(\theta)).$

\item \label{itm:no-value}
  If $(1 - \pv)^2 + \pu\leq (\mathbb{E}_μ (\theta) - \pv)^2$, then
  the optimal veto mechanism is characterized by 
  $\amap^*(θ) = \mathbb{E}_μ (\theta)$ and
  $\pmap^*(θ) = 0$ for all $θ \in Θ$;
  that is, the principal chooses
  $\mathbb{E}_μ (\theta)$ regardless of the state.
\end{enumerate}
\end{proposition}

\begin{proof}
See \Cref{app-a2}.  
\end{proof}

{Since the principal can always choose his ex ante favorite
option by trivially
delegating a singleton set, %\revise{Reviewer2P3}
he strictly benefits from
eliciting the agent's private information
when a non-trivial mechanism is optimal, as in
case \ref{itm:non-trivial} of \Cref{prp-char}.}
\Cref{fig-optim-separating} illustrates the optimal veto mechanism
in this case:
the principal's choices are pooled at the deterministic action $\bar a$ when the state is below the threshold $\bar \theta$;
when the state is above the threshold $\bar \theta$,
the principal uses veto with probability 
$\frac{\eta(\theta) - \bar a}{\eta(\theta) - \pv}$ and chooses 
$\eta(\theta)$ with the complementary probability.
In case~\ref{itm:no-value}, the principal trivially delegates $\set{\E_μ(θ)}$.

\begin{figure}
\centering
\begin{subfigure}{0.4\textwidth} \centering
  \begin{tikzpicture}[domain=0.001:1, scale=4, xscale=1,every node/.style={scale = 0.8}] 
  %% vertices
    \draw[->,] (-0.02,0) node[left]{$0$} -- (1.1,0) node[right] {$\theta$}; 
    \draw[->,] (0,-0.1) -- (0,1.1) node[above] {$\amap(\theta)$};
    %% draw the curve
    \draw[very thick, domain=0.48:1]   plot (\x,{ ((\x + 1) *(\x + 1) -1.4 )^(1/2) - 0.7}); 
    %% \bar a
    \draw[very thick]  (0, 0.19) node [left] {$\bar a$} -- (0.48, 0.19); 
    %% \bar\state
    \draw[dotted, thick] (0.48, 0) node [below] {$\bar\state$} -- (0.48,0.2);
    %% dotted lines
    \draw[dotted,thick] (1,0) -- (1,1);
    \draw[dotted, thick] (0,0) -- (1,1);
    \draw[dotted,thick] (0,1) node [left, yshift = 0.1cm] {$1$}-- (1,1);
    %% 1
    \draw [thick] (1, 0) node [below] {1}-- (1, 0.02);
  \end{tikzpicture}
  \caption*{Action $\amap^*(θ)$}
\end{subfigure} 
\hspace{1.2cm}	
\begin{subfigure}{0.4\textwidth}\centering  
  \begin{tikzpicture}[domain=0.001:1, scale=4, xscale=1, every node/.style={scale = 0.8}] 
  %% vertices
    \draw[->,] (-0.02,0) node[left]{$0$} -- (1.1,0) node[right] {$\theta$}; 
    \draw[->,] (0,-0.1) -- (0,1.1) node[above] {$\pmap(\theta)$};
    %% \bar\state
    \draw[dotted,thick] (1,0) -- (1,1);
    \draw[dotted,thick] (0,1) node [left] {$1$}-- (1,1);
    %% 1
    \draw [] (1, 0) node [below] {1}-- (1, 0.02);
    %% p(\theta, v^*)
    \draw[very thick, domain=0.33066:1]   plot (\x,{1 - 0.7/(sqrt((\x + 0.5) *(\x + 0.5) -0.2 ))}); 
    \draw[very thick]  (0, 0) -- (0.33066, 0) node [below]{$\bar\state$};
  \end{tikzpicture}
  \caption*{Veto probability $\pmap^*(θ)$} 
\end{subfigure}
\caption{\label{fig-optim-separating}Optimal veto mechanism}
%\label{fig:semi-pooling}
\end{figure}

% \paragraph{Implementations} 
One interpretation of the non-trivial optimal veto mechanism
is through delegating lotteries. The principal
delegates a set of lotteries
$\{ l_\theta \}_{ \theta \in [{\ttheta},1] }$ to the agent,
indexed by $\theta$, where 
\(
l_\theta = \Big(\pmap^*(\theta) \circ a_0, (1 - \pmap^*(\theta)) \circ
\eta(\theta) \Big) \text{ for } \theta \in [\bar \theta,1].
\)
The agent chooses the lottery $l_\state$ at $θ \in [\bar \theta,1]$,
and chooses $l_{\bar \theta}$ otherwise.
Another interpretation is through the veto delegation arrangement,
where the agent proposes an option and the principal either vetoes or approves.
Specifically, the principal pre-commits to the following stochastic choice
contingent on the agent's proposal:
the principal uses veto with probability $\pmap^*(\eta ^{-1}(\hat a) )$
if the proposed option $\hat a \in [\bar a, \eta(1)]$,
always approves if $\hat a \in [0, \bar a]$
and otherwise always uses veto.

\Cref{fig-optim-separating} suggests that
both $\amap^*(\theta)$ and $\pmap^*(\theta)$ are increasing over $[\ttheta, 1]$ and, 
perhaps surprisingly, that action $\amap^*(\theta)$ is strictly lower than $\theta$ for 
all $\theta > \bar {\theta}$.
We summarize these properties
in \Cref{cor-valuable-mechanism},
which can be verified directly 
from the corresponding expressions.

\begin{corollary}\label{cor-valuable-mechanism}

\niu{Among all valuable optimal veto mechanisms,}
\begin{enumerate}[(a)]
\item \label{itm:cor-a-p}
  $\amap^*(\theta)$ and $\pmap^*(\theta)$ are weakly increasing in $\theta$;
\item \label{itm:cor-eta}
  $\amap^*(\theta)<\theta$ for all $\theta > \bar {\theta}$.
\end{enumerate}
\end{corollary}

\begin{proof}
  See \Cref{app-a3}.
\end{proof}

The intuition of Corollary~\ref{cor-valuable-mechanism}\ref{itm:cor-a-p} 
is as follows. 
As the state $\theta$ increases, 
the principal's preferred option gets higher
and thus $\amap^*(θ)$ is weakly increasing. 
On the other hand, to maintain the
indifference constraints \eqref{eq-ic}, 
the principal has to veto
those higher options with higher probabilities.
So $\pmap^*(\cdot)$ is weakly increasing as well.
\begin{comment}
This explains the separating part. The pooling part comes from the fact that the
probabilities cannot be negative. When probabilities reaches $0$ and
cannot be adjusted, the actions also need to maintain at the same level
$\bar {a}$.    
\end{comment}
\Cref{cor-valuable-mechanism}\ref{itm:cor-eta} claims that
in the optimal veto mechanism
the proposed option is always lower than the state,
\begin{comment}
The principal obtains the highest payoff when the action
matches with the state $\theta$, while
Corollary~\ref{cor-valuable-mechanism}\ref{itm:cor-eta} claims
that the principal would be better off by choosing a lower action.   
\end{comment}
and its intuition is as follows. First let 
$\amap(\theta) = \theta$ for $θ > \bar \state$ and 
consider a marginal decrease of the proposed option. 
The principal's utility conditional on not using
veto decreases, but that loss is of 
second-order.\footnotemark{} 
On the other hand,
constraints~\eqref{eq-ic}
imply that a lower proposal necessitates a
lower veto probability.
As the status quo option yields payoffs $\pu < 0$ for the principal,
the gain from a marginal downward deviation from the state-matching action is first-order. 
To sum up, it benefits the principal to design the
options $\amap(θ)$ lower than $\theta$ for $θ > \bar \state$.

\footnotetext{As the principal's utility function is quadratic for $a \in [-M,M]$, the first-order derivative with respect to $a$ at $a=\state$ is zero: $\frac{\partial u(a,θ)}{\partial a} \big| _{a = \state} = 0$ for all $θ > \bar \state$.}

\subsection{When are veto mechanisms valuable?}

We say a veto mechanism is \textit{valuable}
if it yields a higher payoff for the principal than that of
delegating $\set{\E_μ(θ)}$.
By \Cref{prp-char}, 
whether there exists a valuable
veto mechanism
depends on players' payoffs from the status quo option,
$\un$ and $\vn$, and the mean of prior $μ$.
\begin{comment}
As $\vn \le 0$, the status quo option can serve as a punishment option to balance the constraints~\eqref{eq-ic}.  
\end{comment}

\paragraph{Effects of $\un$}
As $\vn \le 0$, the status quo option $\an$ serves as a ``punishment device.''
That is, the principal uses the status quo option to balance
the incentives of the agent to select higher options. 
\Cref{prp-char} implies that 
the optimal veto mechanism is valuable
if and only if $\un > (\E _μ(θ) - \vn)^2 - (1-\vn)^2$.
Put in other words,
the status quo option serves as a valid 
punishment option only if it does not harm the
principal that much.
Indeed, when $\pu \le 2v_0 - 1$, 
there does not exist a valuable
veto mechanism no matter what the principal's prior belief is.
On the other hand,
the restriction on $\un$ for
the existence of valuable veto mechanisms
is not severe.
It can be verified that
$\an$ can serve as a valid
punishment option even when $\un$
is significantly less than the principal's payoff
from choosing his ex-ante preferred 
option $\E_μ(θ)$.\footnote{
  For an illustration, let the prior belief be uniform
  over $Θ$ and set $\vn=0$. Then the status quo
  option can serve as a valid punishment option as long as $\un$ is higher than $-3/4$, which is
  less than $-1/12$, the principal's payoff from choosing $\E_μ(\theta)$.
}

\begin{comment}
The veto interpretation implies that in order to maintain the
constraints \eqref{eq-ic}, the principal
punishes the agent for a high proposal by increasing the veto probability
$\pmap(\theta)$.

However, this act of punishment may backfire when $\pu$
is sufficiently low. 
\Cref{prp-char}(i) confirms this
intuition. 

\niu{The no-communication result is due to the agent’s transparent motive. In this setting, the most effective way to induce information provision from the agent is to include the status quo option in the support of the delegated lotteries (\Cref{veto-is-optimal}).}
Consequently, the principal is reluctant
to induce communication when the payoff of the status quo option is
sufficiently small for him.    
\end{comment}

\paragraph{Effects of $\vn$}
The lower bound of $\pu$
for the existence of a valuable veto mechanism, 
$(\E _μ(θ) - \vn)^2 - (1-\vn)^2$, 
is increasing in $\pv$.
% Therefore, a higher $\pv$ makes valuable veto mechanisms  harder to sustain.
The intuition is that
when $\pv$ gets larger, 
the degree of punishment by veto becomes
less severe for the agent. 
Then the principal has to punish the agent with higher
probabilities in order to maintain the constraints~\eqref{eq-ic}. 
Therefore, the principal obtains the status quo
payoff more often, making the veto mechanism 
harder to sustain.

\paragraph{Effects of prior $μ$}
The optimal veto mechanism is valuable only if
$\pu$ is above the threshold $-(1 - \pv)^2$.
In that case, the
necessary and sufficient condition for the existence of 
valuable veto mechanisms can be written more
parsimoniously as
$\mathbb{E}_μ (\theta) < \eta(1) \equiv \sqrt{(1-\vn)^2+\un} + \vn$.
%
\begin{comment}
prerequisites in cases \ref{itm:non-trivial} and \ref{itm:no-value} of \Cref{prp-char} can be written more
parsimoniously as $\eta(1) \leq \mathbb{E}_μ (\theta)$ and
$\eta(1) > \mathbb{E}_μ (\theta)$ respectively.   
\end{comment}
%
Fixing $\pu$ and $\pv$ (and hence $\eta(1)$),
the existence of valuable veto mechanisms
depends only on the mean of the prior $μ$. 
As the principal's ex-post preferred option is $a = \state$
and the agent always prefers higher actions regardless of the state, a higher $\mathbb{E}_μ(\theta)$
indicates that the two players' interests are more 
aligned ex ante. 
In this sense, \Cref{prp-char} implies that 
veto mechanisms are not valuable when
the interests of two players are sufficiently aligned 
(i.e., $\E_μ(\theta) > \eta(1)$).
This finding is in contrast with the \textit{Ally Principle} that
greater alignment leads to more discretion in the
delegation set.%
%%%%
  \footnote{The Ally Principle 
  has been shown to hold in a number of models in the political economy literature (See Bendor, Glazer and Hammond~(2001) and the discussions therein).
  Holmstrom~(1980) demonstrates that
  the Ally Principle holds under general conditions when restricting to interval delegations.}
%%%%
In our scenario, the principal leaves no discretion for the agent when players' preferences are sufficiently aligned, but leaves some discretion when preferences are sufficiently misaligned.

In a related paper, Kartik, Kleiner and Van Weelden~(2021) document
the invalidity of Ally Principle in a bargaining
environment---a proposer (principal) delegates 
a set of lotteries for a vetoer (agent) to choose from
while the vetoer can always choose to maintain the status quo besides the delegated options.
In Kartik et al.~(2021), the opposite of the Ally Principle is true:
greater ex-ante alignment tends to make the principal leave less discretion for the agent (i.e., the delegation set gets strictly smaller).
In a standard optimal delegation setting,
Alonso and Matouschek~(2008, Section~6.4)
illustrate that the Ally Principle
may fail when the optimal deterministic mechanism is not interval delegation.
In this paper, the Ally Principle in general does not hold---the delegation set gets neither strictly smaller nor bigger as players have greater ex-ante preference alignment.\footnote{It is worth noting that 
in most papers where the Ally Principle holds,
the degree of preferences misalignment is measured by some bias parameter and is constant across different states. In contrast, 
both Kartik et al.~(2021) 
and our paper assume one player
having state-independent preferences, 
and thus the degree of (ex-ante)
preferences alignment is belief-dependent. 
In the counter-example of Alonso and Matouschek~(2008), 
players' preference misalignment is not constant across states and is larger at higher states.}

% a similar result with ours 
%that greater ex-ante alignment can make the principal leave less discretion for the agent.
%Specifically, 
%in KKVW's setting the opposite of the Ally Principle is true:
%greater ex-ante alignment tends to make the principal leave less discretion for the agent (i.e., the delegation set gets strictly smaller). In our setting, the delegation set gets neither strictly smaller nor bigger.
%KKVW study the delegation problem in a bargaining

%environment as opposed to the expertise-based delegation, and suggest
%that the violation of Ally Principle should be owed to the
%different rationals for discretion. Our result shows that the
%Ally Principle can fail in an expertise-based delegation model as well
%when the agent's preferences are state-independent. 

%\footnotetext{Specifically, 
%  in KKVW's setting the opposite of the Ally Principle is true:
%  greater ex-ante alignment tends to make the principal leave less discretion for the agent (i.e., the delegation set gets strictly smaller). In our setting, the delegation set gets neither strictly smaller nor bigger.}

A deeper understanding of the effects of $μ$ can be obtained by 
closely examining the threshold $\ta$.
The first-order condition determining $\ta$ permits a
graphical expression as in \Cref{fig-foc}, where
$\ta$ and $\ts \equiv \eta^{-1}(\bar a)$ are the exact
values that equalize
the areas of the two shaded regions
weighted by distribution $μ$. 
Since $\eta(1)<1$,  
when the prior distribution $μ$ puts too much mass on the 
interval $(1-\epsilon,1)$ for
a sufficiently small $\epsilon>0$,
$\bar a$ would be above $\eta(1)$
and then the principal delegates $\{\E_μ(θ)\}$. Therefore, when
$\mathbb{E}_μ(\theta)$ is sufficiently high, 
there will be no valuable veto mechanisms.
\begin{figure}
\begin{center}
  \begin{tikzpicture}[domain=0.001:1, scale=4.5, xscale=1,every node/.style={scale = 0.8}] 
  	%% fill color    	
  	\fill[color=gray!30] (0, 0) -- (0.2,0.2) --(0, 0.2) -- (0,0);
  	\fill[color=gray!30] (0.2,0.2) -- (0.48,0.2) -- (1, 0.915) -- (1,1) -- (0.2,0.2);
  %% vertices
    \draw[->,] (-0.02,0) node[left]{$0$} -- (1.1,0) node[right] {$\theta$}; 
    \draw[->,] (0,-0.1) -- (0,1.1) node[above] {$\amap(\theta)$};
    %% draw the curve
    \draw[very thick, domain=0.48:1]   plot (\x,{ ((\x + 1) *(\x + 1) -1.4 )^(1/2) - 0.7}); 
    %% \bar a
    \draw[very thick]  (0, 0.19) node [left] {$\bar a$} -- (0.48, 0.19); 
    %% \bar\state
    \draw[dotted, thick] (0.48, 0) node [below] {$\bar\state$} -- (0.48,0.2);
    %% dotted lines
    \draw[dotted,thick] (1,0) -- (1,1);
    \draw[dotted, thick] (0,0) -- (1,1);
    \draw[dotted,thick] (0,1) node [left, yshift = 0.1cm] {$1$}-- (1,1);
    %% 1
    \draw [thick] (1, 0) node [below] {1}-- (1, 0.02);
    %% \eta(1)
    \draw[dotted,] (0,0.915) node[left]{
		$\eta(1)$} -- (1,0.915) node[right] {}; 
    \draw (0.7,0.5) node [right] {$\eta(\theta)$};
  \end{tikzpicture}
\end{center}
\caption{\label{fig-foc}Pinning down the thresholds $\bar a$ and $\bar {\theta}$}
\end{figure}

\subsection{Comparative statics}\label{sec-compa}
We derive two comparative statics to demonstrate how changes in $\un$ and $\vn$ affect the optimal veto mechanism,
focusing exclusively on valuable optimal veto mechanisms.

\begin{proposition}\label{prp-cs-u}
Among valuable optimal veto mechanisms, 
if the value of status quo option for the principal $\pu$ is lower, 
then
\begin{enumerate}[(a)]
\item
  $\eta (\theta)$ decreases for each possible
  $\theta$;
\item
  $\pmap^*(\theta)$ decreases whenever 
  $\pmap^*(\theta) > 0$;
\item
  both $\bar {a}$ and $\bar {\theta}$ increase.
\end{enumerate}

\end{proposition}

\begin{proof}
  See \Cref{app-a5}.
\end{proof}

\begin{figure}
  \begin{center}
    \begin{tikzpicture}[domain=0.001:1, scale=4, xscale=1,
                      every node/.style={scale = 0.8}] 
      %% vertices
      \draw[->, ] (0,0) node[left]{$0$} -- (1.1,0) node[right] {$\theta$}; 
      \draw[->, ] (0,0) -- (0,1.1) node[above] {$\amap(\theta)$};
      %% a(\theta) v = - 0.6, u = -0.2 
  %    \draw[,domain=0.33066:1]   plot (\x,{sqrt((\x + 0.6) *(\x + 0.6) -0.2 ) -0.6}); 
      %% some artificial curve
       \draw[domain=0.48:1]   plot (\x,{ ((\x + 1) *(\x + 1) -1.3 )^(1/2) - 0.7});
      %% a(\theta) v = -0.6, u = -0.1
      \draw[dashed, very thick, , domain=0.22:1]   plot (\x,{sqrt((\x + 0.6) *(\x + 0.6) -0.1 ) -0.6});  
      %% v^*
      \draw[]  (0, 0.24) node [left, yshift=0.15cm] {$\color{black} \bar a'$} -- (0.48, 0.24); 
      \draw[dashed, very thick,  ]  (0, 0.1566) node [left, yshift = -0.15cm] {$\color{black} \bar a$} -- (0.22, 0.1566); 
      %% \ts
      \draw[dotted, ] (0.48, 0) node [below] {$\bar \theta'$} -- (0.48,0.24);
      \draw[dotted, ] (0.22, 0) node [below, xshift=-0.1cm] {$\bar \theta$} -- (0.22,0.1566);
      %% v_r
      %     	\draw[dotted, ] (0,0.93178) node[left, yshift = -0.1cm] {$ \eta (1)$} -- (1,0.93178); 
      %% dotted lines
      \draw[dotted, ] (1,0) -- (1,1);
      %\draw[dotted,  ] (0,0) -- (1,1);
      \draw[dotted, ] (0,1) node [left, yshift = 0.1cm] {$1$}-- (1,1);
      %% 1
      \draw [ ] (1, 0) node [below] {1}-- (1, 0.02);
      %% arrows
      %\draw [->,   ,brown] (-0.18,0.22) -- (-0.18,0.13);
      %\draw [->,  , brown] (0.31, -0.18) -- (0.2, -0.18);
      \node [right] at (0,-0.3)  {(a) Action $\amap^*(\theta)$};
    \end{tikzpicture}    
    \qquad % <----------------- SPACE BETWEEN PICTURES
    \begin{tikzpicture}[domain=0.001:1, scale=4, xscale=1,
                     every node/.style={scale = 0.8}] 
        %% vertices
          \draw[->, ] (0,0) node[left]{$0$} -- (1.1,0) node[right] {$\theta$}; 
        \draw[->, ] (0,0) -- (0,1.1) node[above] {$\pmap(\theta)$};
        %% dotted lines
          \draw[dotted, ] (1,0) -- (1,1);
        \draw[dotted, ] (0,1) node [left] {$1$}-- (1,1);
        %% 1
        \draw [ ] (1, 0) node [below] {1}-- (1, 0.02);
        %% p(\theta, v^*) v = -0.6, u = -0.2
        \draw[, domain=0.33:1]   plot (\x,{1 - (0.215+0.6)/(sqrt((\x + 0.6) *(\x + 0.6) -0.2 ))}); 
        %% p(\theta, v^*) v = -0.6, u = -0.1
        \draw[dashed, very thick  , domain=0.22:1]   plot (\x,{1 - (0.1566+0.6)/(sqrt((\x + 0.6) *(\x + 0.6) -0.1 ))});

        \draw[]  (0, 0) -- (0.33, 0) node [below, ]{$\color{black}\bar \theta'$};
        \draw[dashed, very thick, ]  (0, 0) -- (0.22, 0) node [below, ]{$\color{black}\bar \theta $};
        %% arrows
          \node [right] at (-0.1,-0.3)  {(b) Veto probability $\pmap^*(\theta)$};
    \end{tikzpicture}
  \end{center}
  \caption{Changes of $\amap^*(\theta)$ and $\pmap^*(\theta)$ when $\un$ varies}
  \label{fig-compara-u}
\end{figure}

\Cref{fig-compara-u} illustrates how $\amap^*(\theta)$ and $\pmap^*(\theta)$ change as principal's payoffs from the status quo option vary, with the dashed curves representing the case for a higher $\un$.
The intuition for \Cref{prp-cs-u} is as follows.
With a lower $\pu$,
the principal is less willing to maintain the status quo, 
leading to a decrease in the veto probability $\pmap^* (\theta)$ and an increase in the cutoff $\bar{\theta}$.
To maintain the indifference constraints~\eqref{eq-ic},
the principal lowers the actions $\eta(\theta)$.
The combination of a higher $\bar{\theta}$ and lower $\eta(\theta)$ 
leads to an increase of $\bar a$.

% First, note that to keep the indifference constraints~\eqref{eq-ic}, the principal can either strategically exercise
% his veto power (i.e., choosing a suitable $\pmap(\theta)$) or 
% adjust the action $\eta(\theta)$. 
% Then when $\pu$ increases, the principal is more willing to maintain the status quo and
% the veto probability increases (\Cref{prp-cs-u}(c)). As a substitute,
% the principal increases the action $\eta(\theta)$ (\Cref{prp-cs-u}(a)), which maintains the indifference constraints and increases the principal's payoffs as well. 
% Finally,
% as the status quo option becomes more favorable, the principal is
% more willing to employ the lotteries and thus the cutoff $\bar {\theta}$
% decreases. 
% This effect, combined with the increasing $\eta(\theta)$, leads to
% the decrease of $\bar {a}$ (\Cref{prp-cs-u}(b)).

In \Cref{prp-cs-v}, we derive comparative statics concerning the value of the status quo option for the agent.
Unlike the case with varying $\un$, 
changes of the optimal veto probability 
function $\pmap^*(\theta)$ may or may not be uniform as $\vn$ varies.%
  \footnote{This is illustrated with numerical examples in \Cref{app-a4-numeric}.}

\begin{proposition}\label{prp-cs-v}
  Among valuable optimal veto mechanisms, 
  if the value of status quo option
  for the agent $\pv$ is lower, then
  \begin{enumerate}[(a)]
  \item
    %  separating actions  
    $\eta(\theta)$ increases for each possible $\theta$;
  \item
    both $\bar {a}$ and $\bar {\theta}$ decrease.
  \end{enumerate}
\end{proposition}
  
\begin{proof}
  See \Cref{app-a4}.
\end{proof}
  
\Cref{fig-compara-v} illustrates how $\amap^*(\theta)$ changes as the agent's payoffs from the status quo option vary, with the dashed curve representing the optimal action rule for a higher $\vn$.
% \Cref{fig-compara-v} illustrates how $\amap^*(\theta)$ changes as agent's payoffs from the status quo option become higher, where 
% the dashed curve denotes the optimal action rule with a higher $\vn$.
Intuitively, lower $\pv$ tends to lead to 
lower agent equilibrium payoff $\bar{a}$ as 
the agent's payoffs from the status quo option get lower.
Meanwhile, lower 
$\pv$ implies that punishment by veto becomes more severe for the agent. 
To maintain the indifference constraints~\eqref{eq-ic}, 
the principal tends to increase 
the separating part of optimal actions $\eta (\theta)$.
Since $\bar a$ decreases and $\eta(\theta)$ increases, the cutoff state
$\bar {\theta}\equiv \eta^{-1}(\bar {a})$ decreases.  
  
\begin{figure}
  \begin{center}
  
  \begin{tikzpicture}[domain=0.001:1, scale=4.5, xscale=1, every node/.style={scale = 0.8}] 
        %% vertices
        \draw[->] (0,0) node[left]{$0$} -- (1.1,0) node[right] {$\theta$}; 
        \draw[-> ] (0,0) -- (0,1.1) node[above] {$\amap(\theta)$};
        %% a(\theta) v = - 0.6
        %\draw[dashed,   domain=0.1:1]   plot (\x,{sqrt((\x + 0.6) *(\x + 0.6) -0.2 ) -0.6}); 
        \draw[color=black,  thick,   domain=0.33066:1]   plot (\x,{sqrt((\x + 0.6) *(\x + 0.6) -0.2 ) -0.6}); 
        %% a(\theta) v = -0.3
        %\draw[dashed,   domain=0.2:1]   plot (\x,{sqrt((\x + 0.3) *(\x + 0.3) -0.2 ) -0.3}); 
      %	\draw[dashed, very thick,   domain=0.42:1]   plot (\x,{sqrt((\x + 0.3) *(\x + 0.3) -0.2 ) -0.3}); 
          \draw[dashed, very thick, domain=0.55:1]   plot (\x,{ ((\x + 1) *(\x + 1) -1.4 )^(1/2) - 0.7}); 
        %% v^*
        \draw[color=black,  thick]  (0, 0.215) node [left, yshift =-0.1cm] {$\color{black}\bar a'$} -- (0.33, 0.215); 
        \draw[dashed, very thick]  (0, 0.31) node [left, yshift =0.1cm] {$\color{black}\bar a$} -- (0.55, 0.31); 
        %% \ts
        \draw[dotted ] (0.33, 0) node [below] {$\bar \theta'$} -- (0.33,0.215);
        \draw[dotted ] (0.55, 0) node [below, xshift=0.1cm] {$\bar \theta$} -- (0.55,0.32);
        %% v_r
  %     	\draw[dotted ] (0,0.93178) node[left, yshift = -0.1cm] {$ \eta (1)$} -- (1,0.93178); 
          %% dotted lines
        \draw[dotted ] (1,0) -- (1,1);
        %\draw[dotted] (0,0) -- (1,1);
        \draw[dotted] (0,1) node [left, yshift = 0.1cm] {$1$}-- (1,1);
        %% 1
        \draw [] (1, 0) node [below] {1}-- (1, 0.02);
        %% arrows
        %\draw [->, thick ,brown] (-0.18,0.2) -- (-0.18,0.28);
        %\draw [->,   brown] (0.3, -0.18) -- (0.42, -0.18);
  \end{tikzpicture}
  
  \end{center}
  \caption{\label{fig-compara-v}
  Changes of $\amap^*(\theta)$ when $\pv$ varies}
\end{figure}

Lastly, we briefly discuss the comparative statics regarding 
the prior belief.
One might conjecture that, among valuable optimal veto 
mechanisms, the state cutoff $\ts$ changes monotonically
when usual stochastic orderings 
(such as first-order stochastic dominance and likelihood ratios) 
are imposed on the prior beliefs.
Nevertheless, this does not hold in our setting. 
%The reason is that when
%$\theta$ lies in $(\bar {\theta}, 1)$, the distortion
%$\theta- \eta(\theta)$ is decreasing and convex. 
In \Cref{app-a6}, we provide numeric examples demonstrating that
when $\mu_1$ dominates $\mu_2$
in the sense of likelihood ratio (and hence first-order stochastic dominance), the state cutoff induced by prior $\mu_1$ can be either higher or lower than that induced by prior $\mu_2$.

\section{Discussions}\label{sec-exten}
In the previous analysis, we have assumed $\pv \le 0$.
In this section, we discuss how a positive $\pv$ will affect the optimal mechanism.
When $\pv \ge 1$, the optimal mechanism resembles that of our main model, 
with the key difference being that $\an$
acts as a reward rather than a punishment option.
However, when $0 < \pv <1$,
the structure of the optimal mechanism may change significantly.

\subsection{Scenario 1: $\pv \ge 1$}
Suppose $\pv \ge 1$. Then the status quo option serves as a reward option
rather than a punishment option.
As \Cref{prp:main} does not depend on the
value of $\pv$, we can still focus on the class of veto mechanisms
and derive for the optimal mechanism 
by solving sequential maximization problems
\eqref{model:dm3} and \eqref{model:dm4} as in the main model. 
The only difference is that now the
agent prefers the status quo option and $\pv$ becomes the
upper bound for the agent's equilibrium payoff, 
whereas in the previous analysis $\pv$ is the lower bound.
The following proposition summarizes the optimal veto mechanism in this case.

\begin{proposition}\label{prp-ex-optimal-mechanism}
Assume $\pv \ge 1$.

\begin{enumerate}[(a)]
\item
  If $( \pv- \mathbb{E}_{\mu} (\theta))^2 < \pv^2 + \pu$, define
  $\varphi (\state) \equiv \pv- \sqrt{(\pv- \state) ^2+ \pu}$ 
  for $\state \le \pv- \sqrt{- \pu}$ and the optimal veto mechanism {$(\amap^*(\theta), \pmap^*(\theta))$} is
  \[
  \begin{aligned}
        \amap^*(\theta) = \begin{cases}
        \varphi(\theta),  & \text{ if } \theta < { \bar \theta}\\
        {\bar a}, & \text{ if } \theta \ge {\bar \theta}
      \end{cases}
   \qquad\qquad %
      \pmap^*(\theta) = \begin{cases}
        \frac{{\bar a}- \varphi (\theta) }{\pv- \varphi (\theta)},  & \text{ if } \theta < {\bar \theta}\\
        0 , & \text{ if } \theta \ge {\bar \theta}
      \end{cases}
  \end{aligned}
  \]
  where $\bar \theta = \varphi^{-1}(\bar a)$  and $\bar a$
  is determined by
  $\int_0^{\varphi^{-1}(\bar a)}  (\varphi (\theta) - \bar a) \, d\mu  + \bar a - \mathbb{E}_{\mu} (\theta)  = 0$.
  Moreover, $\bar a \in ( \mathbb{E}_{\mu}(\theta), 1).$
\item
  If $( \pv- \mathbb{E}_{\mu} (\theta))^2 \geq \pv^2 + \pu$, then the
  optimal veto mechanism is $\amap^*(\theta) = \mathbb E_{\mu} (\theta)$ and
  $\pmap^*(\theta) = 0$; that is, the principal chooses the action
  $\mathbb{E}_{\mu} (\theta)$ regardless of the state.      
\end{enumerate}
\end{proposition}

\begin{proof}
  See \Cref{app-a7}.
\end{proof}

\Cref{fig-semi-sepa-2} illustrates
the valuable optimal veto mechanism when $\vn \ge 1$:
the principal uses veto only when the
state is below the threshold $\bar\state$,
and those corresponding actions $\amap^*(\theta)$ are higher than $\theta$.
Also, as opposed to the case of $\vn \le 0$,
there exists no
valuable veto mechanism
when $\E_{\mu} (\theta)$ is sufficiently low 
(i.e., when players' preferences are sufficiently misaligned ex ante). 

\begin{figure}
\begin{center}  
  \begin{tikzpicture}[domain=0.001:1, scale=4, xscale=1,every node/.style={scale = 0.8}] 
      \draw[->,] (0,0) node[left]{$0$} -- (1.1,0) node[right] {$\theta$}; 
      \draw[->,] (0,0) -- (0,1.1) node[above] {$\amap(\theta)$};
      %% a(\theta)
%      \draw[color=black, very thick,dashed,  domain=0:0.9]   plot (\x,{1.5 - sqrt((1.5 - \x) *(1.5-\x) -0.2 )}); 
      \draw[color=black, very thick, domain=0:0.6]   plot (\x,{1.5 - sqrt((1.5 - \x) *(1.5-\x) -0.2 )}); 
      %% v^*
      \draw[color=black, dotted]  (0, 0.712) node [left] {$\color{black}\ta$} -- (0.6, 0.712); 
      \draw[color=black, very thick,]  (0.6, 0.712) -- (1, 0.712); 
      %% \ts
      \draw[dotted,] (0.6, 0) node [below] {$\bar{\state}$} -- (0.6,0.712);
      %% v_r
      %     	\draw[dotted,] (0,0.93178) node[left, yshift = -0.1cm] {$ \eta (1)$} -- (1,0.93178); 
      	%% dotted lines
      \draw[dotted,] (1,0) -- (1,1);
      \draw[dotted, ] (0,0) -- (1,1);
      \draw[dotted,] (0,1) node [left, yshift = 0.1cm] {$1$}-- (1,1);
      %% 1
      \draw [] (1, 0) node [below] {1}-- (1, 0.02);
      \node [right] at (0,-0.3)  {(a) Action $\amap^*(\theta)$};
  \end{tikzpicture}    
  \qquad % <----------------- SPACE BETWEEN PICTURES
  \begin{tikzpicture}[domain=0.001:1, scale=4, xscale=1,every node/.style={scale = 0.8}] 
    %% vertices
    \draw[->,] (0,0) node[left]{$0$} -- (1.1,0) node[right] {$\theta$}; 
    \draw[->,] (0,0) -- (0,1.1) node[above] {$\pmap(\theta)$};
    %% v^*
    %   	\draw[color=black, very thick,, ]  (0, 0.2) node [left] {v^*$} -- (0.02, 0.2); 
    %% dotted lines
    \draw[dotted,] (1,0) -- (1,1);
    \draw[dotted,] (0,1) node [left] {$1$}-- (1,1);
    %% 1
    \draw [] (1, 0) node [below] {1}-- (1, 0.02);
    %% p(\theta, v^*)
    \draw[color=black, very thick, domain=0:0.6]   plot (\x,{1 - (1.5-0.712)/(sqrt((1.5-\x ) *(1.5 - \x ) -0.2 ))}); 
    \draw[color=black, very thick,]  (0.6, 0) node [below]{$\color{black}\bar{\state}$} -- (1, 0) ;
    %% p(0)
    \node [right] at (0,-0.3)  {(b) Veto probability $\pmap^*(\theta)$};
  \end{tikzpicture}
\end{center}
\caption{Optimal veto mechanism when $\pv \ge 1$}
\label{fig-semi-sepa-2}
\end{figure}

\subsection{Scenario 2: $\pv \in (0,1)$}\label{sec:extension-discuss}
When $\pv\in (0,1)$, 
solving for the optimal
mechanism generally becomes more complicated.
Here, we characterize the optimal veto mechanism in two specific cases,
where the solutions differ significantly from those in \Cref{prp-char,prp-ex-optimal-mechanism}.

Let $\pv= 0.38$, $\pu= -0.1$ and the prior belief be the
uniform distribution over $\Theta$. 
The optimal action function $\amap^*(\theta)$ and the
optimal veto rule $\pmap^*(\theta)$ are U-shaped:
\[
  \begin{split}
  \amap^*(\theta) = 
      \begin{cases}
      \eta (\theta) & \text{ if } \theta \in [0, \underline{\theta});\footnotemark\\
      \bar {a}& \text{ if } \theta \in [\underline{\theta}, \bar{\theta}); \\
      \eta (\theta) & \text{ if } \theta \in [\bar{\theta}, 1],
      \end{cases} 
  \qquad %
     \pmap^* (\theta) = 
      \begin{cases}
      \frac{\eta (\theta) - \bar {a}}{\eta(\theta) - \pv} & \text{ if } \theta \in [0, \underline{\theta});\\
      0 & \text{ if } \theta \in [\underline{\theta}, \bar{\theta}); \\
      \frac{\eta (\theta) - \bar {a}}{\eta(\theta) - \pv} & \text{ if } \theta \in [\bar{\theta}, 1],
      \end{cases}  
  \end{split}
\]
\footnotetext{Substituting the values of $\vn$
    and $\un$ into the expression of
    $\eta(\state)$ as defined in \Cref{prp-char} 
    yields $\eta(\state)= \sqrt{(\theta -0.38) ^2 - 0.1}+ 0.38$ in this case.}%
where the approximate values are given by
$\underline{\theta} \approx 0.063$, 
$\bar{\theta} \approx 0.697$ 
and $\bar{a} \approx 0.397$.
To obtain this result, we first divide the problem into two
categories, $\av \geq \pv$ and $\av \leq \pv$, where $\av$ is the agent's expected payoff. 
For each case, we employ the two-step
method to solve \eqref{model:dm3} and \eqref{model:dm4} sequentially.
We find that it is optimal for the 
principal to set $\hat v \ge \vn$, and 
thus the status quo option serves as
a stick rather than a  carrot.
Detailed derivations are relegated to \Cref{app-a81}.

\Cref{fig-ushape} illustrates the
optimal veto mechanism:
both $\amap^*(\theta)$ and $\pmap^*(\theta)$ are U-shaped and the principal uses veto at both the higher and the lower states.
Note that when $\theta < \underline\state$, the principal's payoffs from those actions $\amap^*(\theta)$
are lower than that from 
the pooling action $\bar a$,
yet the principal finds it optimal
to take these higher actions.
The reason for this phenomenon is as follows.
When $\theta \in [0, \underline{\theta}]$,
principal's payoff from the pooling action $\bar a$ is lower than $\pu$.
It follows that
the principal has incentives to 
put more probability weight on the status quo option. As a result, the
principal must also set 
$\amap(\theta)$ higher than $\bar a$ to maintain the 
indifference constraints \eqref{eq-ic}.  

\begin{comment}
When $\av \leq \pv$, the optimal veto mechanism is obtained at the corner case of $\av = \pv$.
 which is outcome equivalent (in terms of the principal's
expected payoff) to $\av = \pv$ when $\av \geq \pv$ and is thus
subsumed by the case $\av \geq \pv$. 
When $\av \geq \pv$, the optimal $\av$, which equals $\bar a$ in the equilibrium, is given by the FOC equation \eqref{eq-a-ushape} in the appendix.   
\end{comment}

\begin{figure}
\centering

\begin{minipage}[t]{0.50\linewidth}\centering
{ 
  \begin{tikzpicture}[domain=0.001:1, scale=4, xscale=1,
                      every node/.style={scale = 0.8}] 
    %% vertices
    \draw[->, ] (0,0) node[left]{$0$} -- (1.1,0) node[right] {$\theta$}; 
    \draw[->, ] (0,0) -- (0,1.1) node[above] {$\amap(\theta)$};
      %%%%%  
    \draw[very thick,  , domain=0.6967:1]   plot (\x,{sqrt((\x -0.38) *(\x - 0.38) -0.1 ) + 0.38}); 
    \draw[very thick,  , domain=0:0.0634]   plot (\x,{sqrt((\x -0.38) *(\x - 0.38) -0.1 ) + 0.38}); 
    %%
    %% v^*
    \draw[dotted, ]  (0, 0.397) node [left] {$\color{black} \ta$} -- (1, 0.397); 
    \draw[very thick  ]  (0.0634, 0.397) -- (0.6967, 0.397); 
    	%% dotted lines
    \draw[dotted, ] (1,0) -- (1,1);
    \draw[dotted,  ] (0,0) -- (1,1);
    \draw[dotted, ] (0,1) node [left, yshift = 0.1cm] {$1$}-- (1,1);
    %% cutoffs
    \draw[dotted,  ] (0.0634,0) node [below] {$\underline{\theta}$}-- (0.0634,0.397);
    \draw[dotted,  ] (0.6967,0) node [below] {$\bar{\theta}$}-- (0.6967,0.397);
    %% 1
    \draw [ ] (1, 0) node [below, xshift=0.1cm] {1}-- (1, 0.02);
  \end{tikzpicture}
}
\subcaption{\label{fig-surus}Action $\amap^*(\theta)$}
\end{minipage}%
\begin{minipage}[t]{0.50\linewidth}\centering 
{
\begin{tikzpicture}[domain=0.001:1, scale=4, xscale=1,
                    every node/.style={scale = 0.8}] 
  %% vertices
  \draw[->, ] (0,0) node[left]{$0$} -- (1.1,0) node[right] {$\theta$}; 
  \draw[->, ] (0,-0) -- (0,1.1) node[above] {$\pmap(\theta)$};
  	%% dotted lines
  \draw[dotted, ] (1,0) -- (1,1);
  \draw[dotted, ] (0,1) node [left, yshift = 0.1cm] {$1$}-- (1,1);
  %% 1
  \draw [ ] (1, 0) node [below, xshift=0.1cm] {1}-- (1, 0.02);
  %% cutoffs
  \draw [ ,very thick] (0.0634,0) node [below] {$\color{black} \underline{\theta}$} -- (0.6967,0) node [below] {$\color{black} \bar{\theta}$};
  %% p(\theta, v^*) v = 0.38, u = -0.1 \ta = 0.397
  \draw[very thick  , domain=0:0.06330434]   plot (\x,{1 - (0.3969344 -0.38)/(sqrt((\x -0.38) *(\x -0.38) -0.1))}); 
  \draw[very thick ] (0.06330434,0) -- (0.06330434,0.4); 
  \draw[very thick, domain=0.6967:1]   plot (\x,{1 - (0.3969344 -0.38)/(sqrt((\x -0.38) *(\x -0.38) -0.1))}); 
\end{tikzpicture}
}
\subcaption{Optimal veto probability $\pmap^*(\theta)$}
\end{minipage}%
\caption{U-shaped $\amap^*(\theta)$ and $\pmap^*(\theta)$}
\label{fig-ushape}
\end{figure}

% \subsubsection{Optimal veto mechanism can be deterministic}
As a final illustration,
let $\pv= \frac{1}{2}$, $\pu> - \frac{1}{4}$ and the prior belief be the
uniform distribution over $\Theta$. 
In this case, the optimal mechanism is deterministic:
\[
  (\amap^*(\theta) , \pmap^* (\theta)) = 
  \begin{cases}
  % \text{any number } & \text{ when } \theta \in [0, \frac{1}{2} - \sqrt{-\pu})\\
  ( \pv , 0) & \text{ when } \theta \in [\pv - \sqrt{-\pu}, \pv + \sqrt{-\pu}) \\
  (a' , 1)& \text{ otherwise } 
  % \theta \in [\frac{1}{2} + \sqrt{-\pu}, 1]
  \end{cases}
\]
where $a'$ can be any non-status quo option distinct from $\vn$.
In other words, the principal always vetoes unless the 
proposed option is $a=v_0$.
Detailed derivations are relegated to \Cref{app-a82}.

\section{Concluding remarks}\label{sec-concl}

%
%In this paper, we study the optimal delegation problem when the informed agent has
%state-independent preferences. 
%We show that it is without loss of generality to restrict attention to the class of
%veto mechanisms.
%%in which the principal commits to accept the agent's proposal with some probability and reject for the status quo option with the complementary probability.
%Then we characterize the principal's optimal veto
%mechanism and discuss its properties. 
%While we assume a risk-neutral agent in the main model, the optimality of veto mechanisms
%holds when we relax this assumption.
%We also provide comparative statics and discuss some variants of the main model.
This paper is motivated by the observation that both private and public
organizations often use veto delegation to align incentives between principals
and agents. We study an optimal delegation model in which it is optimal for the
principal to use veto to elicit information from the agent.
Our findings have implications for corporate governance and legislative studies,
where veto delegation is prevalent.

The optimality of veto mechanisms hinges on two key assumptions:
state-independent agent preferences and the principal being more risk-averse
than the agent towards non-status quo options.
Further investigation into how
the optimality of veto mechanisms depends on 
these conditions could be valuable. 
Additionally, 
we have made simplifying assumptions about players' payoff functions
when characterizing the optimal mechanism.
Relaxing these assumptions and exploring the optimal mechanism 
under different conditions could prove worthwhile for future research.

% \bibliographystyle{jpe}
% \bibliography{lei.bib}
\section*{References}\label{references}

\begin{CSLReferences}{1}{0}
    Aghion, P., and J. Tirole (1997), {``Formal and real authority in
    organizations,''} \emph{Journal of Political Economy}, 105(1), 1--29.
    
    Alonso, R., and N. Matouschek (2008), {``Optimal delegation,''}
    \emph{Review of Economic Studies}, 75(1), 259--293.
    
    Amador, M., and K. Bagwell (2013),
    ``The theory of optimal
    delegation with an application to tariff caps,''
    \emph{Econometrica},
    81(4), 1541--1599.

    Baron, D.P., and R.B. Myerson (1982), {``Regulating a monopolist with
    unknown costs,''} \emph{Econometrica}, 50(4), 911--930.
    
    Bartling, B., and Fischbacher, U. (2012).
    ``Shifting the blame: On delegation and responsibility.'' 
    \emph{Review of Economic Studies}, 79(1), 67--87.
    
    Bendor, J., A. Glazer, and T. Hammond (2001), {``Theories of
    delegation,''} \emph{Annual Review of Political Science}, 4(1),
    235--269.

    Chakraborty, A., and R. Harbaugh (2010), {``Persuasion by cheap talk,''}
    \emph{American Economic Review}, 100(5), 2361--82.
    
    Chen, Y. (2022), ``Dynamic delegation with a persistent state,''
    \emph{Theoretical Economics}, 17, 1589--1618.

    Dessein, W. (2002), {``Authority and communication in organizations,''}
    \emph{Review of Economic Studies}, 69(4), 811--838.

    Diehl, C., and C. Kuzmics (2021), {``The (non-)robustness of influential
    cheap talk equilibria when the sender's preferences are state
    independent,''} \emph{International Journal of Game Theory}, 1--15.
    
    Doran, M. (2010), {``The closed rule,''} \emph{Emory Law Journal}, 59, 1363.
    
    Frankel, A. (2016), ``Discounted quotas,''
    \emph{Journal of Economic Theory}, 166, 396--444.

    Gale, D., and M. Hellwig (1985), 
    ``Incentive-compatible debt contracts: The one-period problem,'' 
    \textit{Review of Economic Studies}, 52(4), 647--663.
    
    Gilligan, T. W., and K. Krehbiel (1987), ``Collective decision-making and standing committees: An informational rationale for restrictive amendment procedures,'' \emph{Journal of Law, Economics, and Organization}, 3(2), 287--335.

    ------ (1989), ``Asymmetric information and legislative rules with a heterogeneous committee,'' \emph{American Journal of Political Science}, 33(2), 459--490.
    
    Glazer, J., and A. Rubinstein (2004), {``On optimal rules of
    persuasion,''} \emph{Econometrica}, 72(6), 1715--1736.

    ------ (2006), {``A study in the pragmatics of persuasion: A game
    theoretical approach,''} \emph{Theoretical Economics}, 1, 395--410.

    Hart, S., I. Kremer, and M. Perry (2017), {``Evidence games: Truth and
    commitment,''} \emph{American Economic Review}, 107(3), 690--713.

    Kahneman, D., and A. Tversky (1979),
    ``Prospect theory: An analysis of decision under risk,''
    \textit{Econometrica}, 47(2), 263--292.
    
    Kamenica, E., and M. Gentzkow (2011), {``Bayesian persuasion,''}
    \emph{American Economic Review}, 101(6), 2590--2615.

    Kartik, N., A. Kleiner, and R. Van Weelden (2021),
    ``Delegation in veto bargaining,''
    \emph{American Economic Review}, 111(12), 4046--87.

    Kleiner, Andreas (2022),
    ``Optimal delegation in a multidimensional world,''
    Working Paper.
    \url{https://arxiv.org/abs/2208.11835}.
    
    Kolotilin, A., and A. Zapechelnyuk (2019),
    ``Persuasion meets delegation,''
    Working Paper.
    \url{https://arxiv.org/abs/1902.02628}.
    
    Kováč, E., and T. Mylovanov (2009), {``Stochastic mechanisms in settings
    without monetary transfers: The regular case,''} \emph{Journal of
    Economic Theory}, 144(4), 1373--1395.

    Krishna, V., and J. Morgan (2001), {``Asymmetric information and
    legislative rules: Some amendments,''} \emph{American Political Science
    Review}, 95(2), 435--452.

    Lipnowski, E., and D. Ravid (2020), {``Cheap talk with transparent
    motives,''} \emph{Econometrica}, 88(4), 1631--1660.

    Lubensky, D., and E. Schmidbauer (2018), {``Equilibrium informativeness
    in veto games,''} \emph{Games and Economic Behavior}, 109, 104--125.

    Marino, A.M. (2007), {``Delegation versus veto in organizational games
    of strategic communication,''} \emph{Journal of Public Economic Theory},
    9(6), 979--992.

    Martimort, D., and A. Semenov (2006), {``Continuity in mechanism design
    without transfers,''} \emph{Economics Letters}, 93(2), 182--189.

    Melumad, N.D., and T. Shibano (1991), {``Communication in settings with
    no transfers,''} \emph{RAND Journal of Economics}, 173--198.

    Mylovanov, T. (2008), {``Veto-based delegation,''} \emph{Journal of
    Economic Theory}, 138(1), 297--307.
\end{CSLReferences}
\clearpage

\appendix
\titleformat{\section}{\centering\Large\bfseries}{}{0em}{}
\renewcommand{\thesection}{\Alph{section}}
\renewcommand{\thesubsection}{\Alph{section}\arabic{subsection}}

\section{Appendix}\label{sec-app}  

The Appendix contains the proofs and derivations which have been
omitted from the main text.

\subsection{Proof of \Cref{prp-char}\label{app-a2}}
We first solve the maximization problem~\eqref{model:dm3} at state $\theta$
for some fixed agent utility $\hat{v}$.
Constraints~\eqref{eq-ic} yield:
\begin{equation} \label{eq:eg-prob-function}
  \pmap(\theta) =  \frac{\amap(\theta) - \av}{\amap(\theta) -\pv}.    
\end{equation}
Equation~\eqref{eq:eg-prob-function} implies that
$\pmap(\theta) \le 1$ is equivalent to $\av \ge \pv$ and that
$\pmap(\theta) \ge 0$ is equivalent to $\av \le \amap(\theta)$.
%When the principal chooses the deterministic action $\pac$ at some state $\res$, we obtain $p_{\res} = 1$ and  impose $a_{\res} = 1$\footnote{Fixing any $a_{\res} \in A$ induces an equivalent outcome as implementing $a_{\res}$ is a zero-probability event.} 
%so the fraction in \eqref{eq:eg-prob-function} is well-defined.
Substitute \Cref{eq:eg-prob-function} into principal's utility function,
and problem~\eqref{model:dm3} is reduced to:
\begin{equation} \label{eq:eg-optimization}
\max_{\amap(\theta) \in [-M , M]} \,  \Big(\frac{\amap(\theta) - \av }{\amap(\theta) -\pv} \Big) \pu - 
\Big(\frac{\av - \pv}{\amap(\theta) -\pv}\Big) (\amap(\theta) -\theta)^2  \quad
\text{subject to }  \amap(\theta) \ge \av. 
\end{equation}
% \begin{equation*}
% \begin{split}
%   &\max_{\amap(\theta) \in [-M , M]} \,  \Big(\frac{\amap(\theta) - \av }{\amap(\theta) -\pv} \Big) \pu - 
%   \Big(\frac{\av - \pv}{\amap(\theta) -\pv}\Big) (\amap(\theta) -\theta)^2    \\
%   &\text{ subject to }  \amap(\theta) \ge \av. 
% \end{split}
% \end{equation*}
%where we omit the parameter $\theta$ in $\amap(\theta)$ and
%fix $\av \in [\pv, 1]$. 
%One can interpret $\theta \in [0,1]$ as a fixed parameter in Equation \eqref{eq:eg-optimization}.
The solution to problem \eqref{eq:eg-optimization} is
\begin{equation}\label{eq:optimal-a}
  \amap^*(\theta) = \begin{cases}
  \max \set{\av, \eta(\theta)} 
    & \text{ if } (\theta - \pv)^2  + \pu \ge 0\\
  \av & \text{ otherwise }     
  \end{cases}
\end{equation}
where $\eta(\theta) = \sqrt{(\theta - \pv)^2 + \pu} + \pv$.  
%To sum up, the solution
%$\amap(\theta)^* = \max \set{\av,  \sqrt{\nu(\theta)} + \pv }$
%where $\nu(\theta) = \max \set{0, \pu + (\theta - \pv)^2}$. 

% \hypertarget{step-2-optimal-v.}{%
% \subsubsection*{\texorpdfstring{Step 2: Optimal
% $\av$.}{Step 2: Optimal v\^{}*.}}\label{step-2-optimal-v.}}
Then, we solve for the optimal $\av$ at state $\theta$ given the action rule specified by \Cref{eq:optimal-a}.  
Specifically, %Based on equation \eqref{eq:optimal-a}, 
we derive the optimal $\av$ for different values of $\pu$ and $\pv$ in the following scenarios:
\begin{enumerate}[I.]
\item 
  {$\pu + \pv ^ 2 \ge 0$.} In this scenario, $\pu + (\theta - \pv) ^ 2 \geq 0$ for all $\theta \in \State$.
\item 
  {$\pu + \pv ^ 2 < 0$} and $\pu + (1 - \pv) ^2 > 0$. In this scenario, $\pu + (\theta - \pv)^2 < 0$ for $\theta \in [0, \sqrt{- \pu} + \pv)$ and
  $\pu + (\theta - \pv) ^2 \geq 0$ for $\theta \in [\sqrt{- \pu} + \pv , 1]$.
\item 
  {$\pu + (1 - \pv) ^2  \le 0$.} In this scenario, $\pu + (\theta - \pv)^2 \le 0$ for all $\theta \in \State$.
\end{enumerate}

\subsubsection*{Scenario I: $\pu + \pv^2 \ge 0$}
We first explicitly write out principal's optimal action $\amap^*(\theta)$ specified by 
\Cref{eq:optimal-a} for three different ranges of $\av$ (Figure \ref{fig:a_1}), and 
then find the optimal $\av$ within each range. 
Finally, we compare the solutions in different 
ranges and find the 
global maximum. 
\begin{figure}[!hptb]
  \centering
  \begin{subfigure}{0.3\textwidth}
  \centering                                                           
  {\  
      \begin{tikzpicture}[domain=0.001:1, scale=3.2, xscale=1, every node/.style={scale=0.7}] 
    %% vertices
    \draw[->,thick] (-0.02,0) node[left]{$0$} -- (1.1,0) node[right] {$\theta$}; 
    \draw[->,thick] (0,-0.6) -- (0,1.1) node[above] {$\amap(\theta)$};
    %% \pv
    \draw (-0.02,-0.5) node[left]{$\pv$} -- (0.02,-0.5);
    %% \amap(\theta)
    \draw[thick, domain=0:1]   plot (\x,{sqrt((\x + 0.5) *(\x + 0.5) -0.2 ) -0.5}); 
    \draw (0,-0.26) node[left]{\footnotesize $\eta(0)$};
    %% v_r
%       \draw[dotted,thick] (0,0.93178) node[left, yshift = -0.1cm] {$ \eta (1)$} -- (1,0.93178); 
    %% dotted lines
    \draw[dotted] (1,0) -- (1,1);
    \draw[dotted] (0,0) -- (1,1);
    \draw[dotted] (0,1) node [left, yshift = 0.1cm] {$1$}-- (1,1);
    \draw[dotted] (0,0.9318) node [left, yshift = -0.1cm] {\footnotesize $\eta(1)$}-- (1,0.9318);
    %% 1
    \draw [thick] (1, 0) node [below] {1}-- (1, 0.02);
    %% eta(\theta)
    \draw[->] (0.86,0.6) node[below] {$\eta(\theta)$} -- (0.8,0.7) ;
    \end{tikzpicture}
  } 
    \subcaption{$\av \in [\pv, \eta(0) ]$}
    \end{subfigure} 
  \begin{subfigure}{0.3\textwidth}
  \centering                                                           
  {\  
    \begin{tikzpicture}[domain=0.001:1, scale=3.2, xscale=1, every node/.style={scale=0.7}] 
    %% vertices
    \draw[->,thick] (-0.02,0) node[left]{$0$} -- (1.1,0) node[right] {$\theta$}; 
    \draw[->,thick] (0,-0.6) -- (0,1.1) node[above] {$\amap(\theta)$};
    %% \pv
    \draw (-0.02,-0.5) node[left]{$\pv$} -- (0.02,-0.5);
    %% \amap(\theta)
    \draw (0,-0.26) node[left]{\footnotesize $\eta(0)$};
    \draw[dashed, domain=0:1]   plot (\x,{sqrt((\x + 0.5) *(\x + 0.5) -0.2 ) -0.5}); 
    \draw[thick, domain=0.33066:1]   plot (\x,{sqrt((\x + 0.5) *(\x + 0.5) -0.2 ) -0.5}); 
    %% \av
    \draw[thick]  (0, 0.2) node [left] {$\av$} -- (0.33066, 0.2); 
    %% \theta^*
    \draw[dotted] (0.33066, 0) node [below] {$\eta^{-1}(\av)$} -- (0.33066,0.2);
    %% v_r
%       \draw[dotted,thick] (0,0.93178) node[left, yshift = -0.1cm] {$ \eta (1)$} -- (1,0.93178); 
    %% dotted lines
    \draw[dotted] (1,0) -- (1,1);
    \draw[dotted] (0,0) -- (1,1);
    \draw[dotted] (0,1) node [left, yshift = 0.1cm] {$1$}-- (1,1);
    \draw[dotted] (0,0.9318) node [left, yshift = -0.1cm] {\footnotesize $\eta(1)$}-- (1,0.9318);
    %% 1
    \draw [thick] (1, 0) node [below] {1}-- (1, 0.02);
    %% eta(\theta)
    \draw[->] (0.86,0.6) node[below] {$\eta(\theta)$} -- (0.8,0.7) ;
    \end{tikzpicture}
  } 
    \subcaption{$\av \in [\eta(0), \eta (1) ]$}
    \end{subfigure} 
  \begin{subfigure}{0.3\textwidth}
  \centering                                                            
  {\  
    \begin{tikzpicture}[domain=0.001:1, scale=3.2, xscale=1, every node/.style={scale=0.7}] 
    %% vertices
    \draw[->,thick] (-0.02,0) node[left]{$0$} -- (1.1,0) node[right] {$\theta$}; 
    \draw[->,thick] (0,-0.6) -- (0,1.1) node[above] {$\amap(\theta)$};
    %% \pv
    \draw (-0.02,-0.5) node[left]{$\pv$} -- (0.02,-0.5);
    %% \amap(\theta)
    \draw[dashed, domain=0:1]   plot (\x,{sqrt((\x + 0.5) *(\x + 0.5) -0.2 ) -0.5}); 
    \draw (0,-0.26) node[left]{\footnotesize $\eta(0)$};
    %% \av
    \draw [thick] (0, 0.96) node [left] {$\av$} -- (1, 0.96) ;
    %% v_r
%       \draw[dotted,thick] (0,0.93178) node[left, yshift = -0.1cm] {$ \eta (1)$} -- (1,0.93178); 
    %% dotted lines
    \draw[dotted] (1,0) -- (1,1);
    \draw[dotted] (0,0) -- (1,1);
    \draw[dotted] (0,1) -- (1,1);
    \draw[dotted] (0,0.9318) -- (1,0.9318);
    %% 1
    \draw [thick] (1, 0) node [below] {1}-- (1, 0.02);
    %% eta(\theta)
    \draw[->] (0.86,0.6) node[below] {$\eta(\theta)$} -- (0.8,0.7) ;
    \end{tikzpicture}
  } 
  \caption{$\av \in [\eta(1), 1 ]$}
  \end{subfigure}
\caption{principal's $\amap^*(\theta)$ for different $\hat v$ when $ \pu + \pv^2 \ge 0$}
\label{fig:a_1}
\end{figure}

When $\av \in [\pv, \eta(0) ]$, principal's optimal action is
\(\amap^*(\theta)  = \eta (\theta),\)
and the veto probability is 
\(\pmap^*(\theta) = \frac{\eta(\theta) - \av}{ \eta (\theta) - \pv}.\) 
Principal's optimization problem is reduced to
\[
\max_{\av \ge 0}   \Gamma_1(\av) \equiv 
  \int_{0}^1 \Big( \frac{\eta(\theta) - \av}{\eta (\theta) - \pv} \Big) \pu - \Big( \frac{\av - \pv}{\eta (\theta) - \pv} \Big) (\eta (\theta) -\theta)^2 \dif μ.   
\]
Since 
\(
  \Gamma_1'(\av)  
  %  = & - \int_{0}^1 \Bigg(\frac{\pu}{\sqrt{(\theta -\pv) ^2+ \pu}} +    \frac{(\sqrt{(\theta-\pv) ^2+ \pu}+ \pv-\theta)^2 }{\sqrt{(\theta -\pv) ^2+ \pu}}\Bigg) \dif μ  \\
%          = & - 2 \int_{0}^1 \eta (\theta) - \theta \dif μ  \\
      = 2 \int_{0}^1 (\theta - \pv)- \sqrt{(\theta-\pv) ^2+ \pu}  \dif μ > 0
\)      
for all $\av \in [\pv, \eta(0)]$,
the objective function $\Gamma_1 (\av)$ attains its maximum at $\av^* = \eta (0)$. 

When $\av \in [\eta(0), \eta(1) ]$, principal's optimal action is
\[
  \amap^*(\theta) = \begin{cases}
    \av , &\text{ if } \theta \leq  \eta \inv (\av); \\
    \eta(\theta), &\text{ if } \theta > \eta \inv (\av),
    \end{cases}
\]
and the veto probability is
\[
  \pmap^*(\theta) = \begin{cases}
    0, &\text{ if } \theta \leq \eta \inv (\av); \\
    \frac{\eta(\theta) - \av}{ \eta (\theta) - \pv}, &\text{ if } \theta > \eta \inv (\av).
  \end{cases}
\]
Principal's maximization problem is reduced to
\[
  \begin{split}
   \max_{\av} \Gamma_2(\av) \equiv \int_0^{\eta \inv (\av)} -  (\av-\theta)^2 \dif μ  + 
      \int_{\eta \inv (\av)}^1 \Big[\Big(\frac{\eta(\theta) - \av}{\eta (\theta) - \pv} \Big) \pu -\Big( \frac{\av - \pv}{\eta (\theta) - \pv} \Big) (\eta (\theta) -\theta)^2 \Big]\dif μ  
  \end{split} 
\]
The first-order derivative is
\( 
	\Gamma_2'(\av) =  - 2 \left( \int_{\eta \inv (\av)}^{1} \eta (\theta) - \av \dif μ  + \av - \mathbb{E}_μ (\theta)  \right).
\)
To find the sign of $\Gamma_2'(\av)$, first note that $\Gamma_2'(\av)$ is decreasing:
$\Gamma_2''(\av) =  - 2 \int_0^{\eta \inv (\av)} 1 \dif μ   < 0.$
At the two end points $\eta(0)$ and $\eta(1)$,  we have
$\Gamma_2'(\eta(0)) = - 2 \int_{0}^{1} \eta (\theta) - \theta \dif μ   > 0$
and $\Gamma_2'(\eta(1)) =  - 2 ( \eta(1) - \mathbb{E}_μ (\theta) )$.
Then,
\begin{enumerate}[noitemsep]
  \item when $\eta (1) \le \mathbb{E}_μ (\theta)$, we have $\Gamma_2'(\eta(1)) \geq 0$ and the solution is $\av^* = \eta(1)$;
  \item when $\eta(1) > \mathbb{E}_μ (\theta)$, we have $\Gamma_2'(\eta(1)) < 0$ and 
        the solution is $\av^* = \bar{v}$ where $\bar{v}$ is determined by $\Gamma_2'(\bar{v}) = 0$. 
\end{enumerate}
Furthermore, %We could further show that $\bar{v} \in (0, \mathbb{E}_μ (\theta))$:
\begin{gather*}
  \Gamma_2'(0)  = - 2 \left( \int_{\eta \inv (0)}^{1} \eta (\theta) \dif μ  - \mathbb{E}_μ (\theta) \right) 
  = 2 \left( \int_{\eta \inv (0)}^{1}   \theta -  \eta (\theta)  \dif μ  +  \int_0^{\eta \inv (0)} \theta \dif μ \right) > 0,\\
  \Gamma_2'(\mathbb{E}_μ (\theta))  = - 2 \left( \int_{\eta \inv (\mathbb{E}_μ (\theta))}^{1} \eta (\theta) - \mathbb{E}_μ (\theta) \dif μ   \right) < 0.
\end{gather*}
Since $ \Gamma_2''(\av)  < 0$, we have $\bar{v} \in (0, \mathbb{E}_μ (\theta) )$.
  
Lastly, when $\av \in [\eta(1), 1 ]$, principal's optimal action is
$\amap^*(\theta) = \av$
and the veto probability is $\pmap^*(\theta) = 0$.
Principal's optimization problem is reduced to
$$
\max_{\av} \Gamma_3(\av) \equiv - \int_0^{1}(\av-\theta)^2 \dif μ .
$$
The first-order derivative is
\(
  \Gamma_3'(\av) =  
  %- \int_0^{1}2 (\av-\theta) \dif μ  = 
  -2 \left( \av - \mathbb{E}_μ (\theta) \right).  
\)
Note that $\Gamma_3''(\av) = -2 < 0 $ and
at the two end points we have: 
\[
	\Gamma_3'(1) =  -2 \left( 1 - \mathbb{E}_μ (\theta) \right)   < 0, \quad
    \Gamma_3'( \eta (1)) =  -2 \left( \eta(1) - \mathbb{E}_μ (\theta) \right). 
\]
Then,
\begin{enumerate}[noitemsep]
  \item when $\eta(1) \leq \mathbb{E}_μ (\theta)$, we have $\Gamma_3'( \eta (1)) \geq 0$
  and the solution is $\av^* = \mathbb E_μ (\theta)$; 
  \item when $\eta(1) > \mathbb{E}_μ (\theta)$, we have $\Gamma_3'( \eta (1)) < 0$
  and the solution is  $\av^* = \eta(1)$.
\end{enumerate}

To conclude, when $\eta(1) \leq \mathbb{E}_μ (\theta)$,
the optimal agent utility level within the three ranges are given by:
\[ \av^* = 
\begin{cases}
  \eta(0) & \text{ if }  \av \in [\pv, \eta(0) ] \\
  \eta(1) & \text{ if }  \av \in [\eta(0), \eta(1)  ] \\
  \mathbb E_μ (\theta)  & \text{ if }  \av \in [\eta(1) , 1] 
\end{cases}
\]
Otherwise, 
\[ \av^* = 
\begin{cases}
	\eta(0) & \text{ if }  \av \in [\pv, \eta(0) ] \\
	\bar{v} & \text{ if }  \av \in [\eta(0), \eta(1) ] \\
	\eta(1) & \text{ if }  \av \in [\eta(1)  , 1] 
\end{cases}
\]
where $\bar{v}$ is the solution to $\Gamma_2'(\bar{v}) = 0$. 

Since the ranges overlap at the corner, we
directly compare the solutions in the three ranges. 
The solution to Scenario I is as follows:
\begin{enumerate}
\item 
  If $\eta(1) \leq \mathbb{E}_μ (\theta)$, we have 
  $\av^* = \mathbb{E}_μ (\theta)$ and the optimal action is 
  $\amap^*(\theta) =  \mathbb{E}_μ (\theta)$ for all $θ$. 
\item 
  If $\eta(1) > \mathbb{E}_μ (\theta)$, we have $\av^* = \bar{v}$ where $\bar{v}$ is given by $\Gamma_2'(\bar{v}) = 0$ and the optimal action is as follows: 
  when $\theta \in [0, \eta \inv (\bar{v})]$, $\amap^*(\theta) = \bar{v}$;
  when $\theta \in ( \eta \inv (\bar{v}), 1]$, 
  $\amap^*(\theta) = (\pmap^*(\theta) \circ \an, (1 - \pmap^*(\theta)) \circ \eta(\theta))$ where 
  $\pmap^*(\theta) = \frac{\eta(\theta) - \bar{v}}{\eta(\theta) - \pv}$.
%   \begin{itemize}
%     \item when $\theta \in [0, \eta \inv (\bar{v})]$, the action is pooled at $\amap^*(\theta) = \bar{v}$, and 
%     \item when $\theta \in ( \eta \inv (\bar{v}), 1]$, the default action $\pac$ is chosen with probability $\frac{\eta(\theta) - \bar{v}}{\eta(\theta) - \pv}$ and the action $\eta(\theta)$ is chosen with the complementary probability.
%   \end{itemize}
\end{enumerate}

\subsubsection*{Scenario II: $\pu+ \pv^ 2 < 0$ and $\pu+ (1 - \pv) ^2 > 0$}
\Cref{fig:a_2} illustrates the two possibles
cases of this scenario:
$\av \in [\pv, \eta(1)]$ and $\av \in [\eta(1), 1]$.
The derivation is similar to that of Scenario I 
(and therefore omitted), and the solution is the same as that of Scenario I.

\begin{figure}
  \centering
  \begin{subfigure}{0.4\textwidth}
  \centering                                                           
  {\  
    \begin{tikzpicture}[domain=0.001:1, scale=3.2, xscale=1, every node/.style={scale=0.7}] 
    %% vertices
    \draw[->,thick] (-0.02,0) node[left]{$0$} -- (1.1,0) node[right] {$\theta$}; 
    \draw[->,thick] (0,-0.6) -- (0,1.1) node[above] {$\amap(\theta)$};
    %% \pv
    \draw [dotted] (-0.02,-0.44) node[left]{$\pv$} -- (0.21,-0.44);
    %% \amap(\theta)
    \draw[dashed, domain=0.21:1]   plot (\x,{sqrt((\x + 0.5) *(\x + 0.5) -0.5 ) -0.5}); 
    \draw[thick, domain=0.5:1]   plot (\x,{sqrt((\x + 0.5) *(\x + 0.5) -0.5 ) -0.5}); 
    %% \av
    \draw[thick]  (0, 0.21) node [left] {$\av$} -- (0.5, 0.21); 
    %% \theta^*
    \draw[dotted] (0.5, 0) node [below] {$\eta^{-1}(\av)$} -- (0.5,0.21);
    %% v_r
%       \draw[dotted,thick] (0,0.93178) node[left, yshift = -0.1cm] {$ \eta (1)$} -- (1,0.93178); 
    %% dotted lines
    \draw[dotted] (1,0) -- (1,1);
    \draw[dotted] (0,0) -- (1,1);
    \draw[dotted] (0,1) node [left, yshift = 0.1cm] {$1$}-- (1,1);
    \draw[dotted] (0,0.9318) node [left, yshift = -0.1cm] {\footnotesize $\eta(1)$}-- (1,0.9318);
    %% 1
    \draw [thick] (1, 0) node [below] {1}-- (1, 0.02);
    %% eta(\theta)
    \draw[->] (0.86,0.4) node[below] {$\eta(\theta)$} -- (0.78,0.52) ;
    \end{tikzpicture}
  } 
    \subcaption{$\av \in [\pv, \eta(1)]$}
    \end{subfigure} 
  \begin{subfigure}{0.4\textwidth}
  \centering                                                            
  {\  
    \begin{tikzpicture}[domain=0.001:1, scale=3.2, xscale=1, every node/.style={scale=0.7}] 
    %% vertices
    \draw[->,thick] (-0.02,0) node[left]{$0$} -- (1.1,0) node[right] {$\theta$}; 
    \draw[->,thick] (0,-0.6) -- (0,1.1) node[above] {$\amap(\theta)$};
    %% \pv
    \draw [dotted] (-0.02,-0.44) node[left]{$\pv$} -- (0.21,-0.44);
    %% \amap(\theta)
    \draw[dashed, domain=0.21:1]   plot (\x,{sqrt((\x + 0.5) *(\x + 0.5) -0.5 ) -0.5}); 
    %% \av
    \draw [thick] (0, 0.96) node [left] {$\av$} -- (1, 0.96) ;
    %% v_r
%       \draw[dotted,thick] (0,0.93178) node[left, yshift = -0.1cm] {$ \eta (1)$} -- (1,0.93178); 
    %% dotted lines
    \draw[dotted] (1,0) -- (1,1);
    \draw[dotted] (0,0) -- (1,1);
    \draw[dotted] (0,1) -- (1,1);
    \draw[dotted] (0,0.9318) -- (1,0.9318);
    %% 1
    \draw [thick] (1, 0) node [below] {1}-- (1, 0.02);
    %% eta(\theta)
    \draw[->] (0.86,0.4) node[below] {$\eta(\theta)$} -- (0.78,0.52) ;
    \end{tikzpicture}
  } 
  \caption{$\av \in [\eta(1), 1 ]$}
  \end{subfigure}
  \caption{principal's $\amap^*(\theta)$ for different $\av$ when $ \pu + \pv^2 < 0$ and $\pu + (1-\pv)^2 > 0$}
  \label{fig:a_2}
\end{figure}
% Since the solution of Scenario I and Scenario II are the same, the two cases are combined and the corresponding condition becomes $\pu + (1 - \pv) ^2 > 0$.
% Furthermore, the condition for pooling action  $\eta(1) \leq \mathbb{E}_μ (\theta)$ is equivalent to $\pu + (1 - \pv) ^2  \leq (\mathbb{E}_μ (\theta) - \pv)^2$.
% Therefore, 
% \begin{itemize}
% \item 
%   When $ 0 < \pu + (1 - \pv) ^2  \leq (\mathbb{E}_μ (\theta) - \pv)^2$, the action is pooled at $\amap^*(\theta) = \mathbb{E}_μ (\theta)$;
% \item 
%   When $\pu + (1 - \pv) ^2 > (\mathbb{E}_μ (\theta) - \pv)^2$, the optimal action rule is the same as equations
%   \eqref{eq:optim-veto} in \Cref{prp-char}.  
% \end{itemize}

\subsubsection*{Scenario III: $\pu + (1 - \pv) ^2 \le 0$}
In this scenario, principal's optimal action as described in
\Cref{eq:optimal-a} is reduced to
\(\amap^*(\theta) = \av \text{ for all } \theta \in [0,1],\)
and the veto probability is $\pmap^*(\theta )= 0$. 
Principal's maximization problem is reduced to
\[ \max_{\av}  - \int_0^{1}(\av-\theta)^2 \dif μ  \]
The first-order condition yields:
\[- \int_0^{1}2 (\av-\theta) \dif μ  = 0 \implies \av^* = \mathbb{E}_μ (\theta) \]
So, the optimal action is $\amap^*(\theta) =  \mathbb{E}_μ (\theta)$ for all $θ$.
%the principal trivially delegates $\set{\E_μ (\theta)}$.

\Cref{prp-char} summarizes the three scenarios.
Scenarios I and II are combined as they 
yield the same solution, 
with the overall prerequisite being
$\pu + (1 - \pv) ^2 > 0$.
Additionally, the condition  $\eta(1) \leq \mathbb{E}_μ (\theta)$ is equivalent to $\pu + (1 - \pv) ^2  \leq (\mathbb{E}_μ (\theta) - \pv)^2$.

% Therefore, 
% \begin{itemize}
% \item 
%   When $ 0 < \pu + (1 - \pv) ^2  \leq (\mathbb{E}_μ (\theta) - \pv)^2$, the action is pooled at $\amap^*(\theta) = \mathbb{E}_μ (\theta)$;
% \item 
%   When $\pu + (1 - \pv) ^2 > (\mathbb{E}_μ (\theta) - \pv)^2$, the optimal action rule is the same as equations
%   \eqref{eq:optim-veto} in \Cref{prp-char}.  
% \end{itemize}

% Proposition \ref{prp-char} summarizes the result. 
% In particular, 
% the case $ (\mathbb{E}_μ (\theta) - \pv)^2 < \pu + (1 - \pv) ^2  $ corresponds to part (a) of the proposition;
% the cases $0 < \pu + (1 - \pv) ^2  \leq (\mathbb{E}_μ (\theta) - \pv)^2$ and $\pu + (1 - \pv) ^2 \le 0$ 
% both result in pooled action and are combined, forming part (b) of the proposition.
% For expositional purpose, we use $\bar a$ to denote $\bar{v}$ in the proposition.

\subsection{Proof of Corollary~\ref{cor-valuable-mechanism}}\label{app-a3}
Suppose the optimal veto mechanism is valuable. 
Then it has the semi-separating form as below:
%%%%%%%%%%%%%
\begin{enumerate}
\def\labelenumi{\arabic{enumi}.}
\item
  When $\theta \in [0, \eta ^{-1}(\ta)]$,
  $\amap^*(\theta) = \ta$
  and $\pmap^*(\theta) = 0$.
\item
  When $\theta \in (\eta ^{-1}(\ta) , 1]$,
  $\amap^*(\theta) = \eta (\theta) = \sqrt{(\theta -\pv) ^2+ \pu}+ \pv$.
  Taking derivative with respect to $\theta$ yields:
  \[\eta'(\theta) = \frac{\theta - \pv}{\sqrt{(\theta - \pv)^2} + \pu} > 0.\]
  % $\pmap(\theta) = \frac{\eta(\theta) - \ta}{\eta(\theta) - \pv}.$
  Taking derivative of $\pmap^*(\theta) = \frac{\eta(\theta) - \ta}{\eta(\theta) - \pv}$ with respect to $\theta$ gives
  \[
  \begin{aligned}
    {{\pmap}^*} {}^\prime (\theta) = \frac{\eta'(\theta) (\eta (\theta) - \pv) - \eta'(\theta)(\eta (\theta) - \ta) }{(\eta (\theta) - \pv)^2} 
    = \frac{\eta'(\theta) (\ta- \pv) }{(\eta (\theta) - \pv)^2}
    > 0.
  \end{aligned}
  \]
\end{enumerate}
Therefore, $\amap^*(\theta)$ and $\pmap^*(\theta)$ are constant over
$\theta \in [0, \eta ^{-1}(\ta)]$ and are strictly increasing over
$\theta \in (\eta ^{-1}(\ta) , 1]$.

As for part (b), fix some $\theta \in (\eta ^{-1}(\ta) , 1]$. Since
$\pu< 0$, we have $\sqrt{(\theta -\pv) ^2+ \pu} < \theta - \pv.$
Therefore,
$\eta (\theta) = \sqrt{(\theta -\pv) ^2+ \pu}+ \pv< \theta$.

\subsection{Proof of Proposition~\ref{prp-cs-u}} \label{app-a5}

Recall $\eta (\theta) = \sqrt{\pu+ (\theta - \pv)^2} + \pv$. Partial
derivative of $\eta (\theta)$ with respect to $\pu$ gives
\[\frac{\partial \eta (\theta)}{\partial \pu} = \frac{1}{2\sqrt{\pu+ (\theta - \pv)^2}}
        > 0 \text{ for all } \theta \in [0,1].\] Therefore,
$\eta (\theta)$ increases with $\pu$ for all $\theta \in [0,1]$.

For $\ta$, recall that $\ta$ is obtained by:
\[\int_{\eta ^{-1}(\ta)}^{1} \eta (\theta) - \ta\, dμ + \ta- \mathbb{E} (\theta) = 0.\]
Taking partial derivative with respect to $\pu$ yields:
$$
\int_{\eta ^{-1}(\ta)}^{1} \frac{\partial \eta (\theta)}{\partial \pu}  \, dμ + \int_0^{\eta ^{-1}(\ta)} \frac{\partial  \ta}{\partial \pu} \, dμ= 0.
$$
% \[
% \begin{aligned}
%     & \frac{\partial \eta ^{-1}(\ta)}{\partial \pu} (\eta (\hat{\theta}) - \ta) g(\hat{\theta}) + \int_{\eta ^{-1}(\ta)}^{1} \frac{\partial \eta (\theta)}{\partial \pu} - \frac{\partial  \ta}{\partial \pu} \, dμ + \frac{\partial  \ta}{\partial \pu} = 0\\
%     \implies & \int_{\eta ^{-1}(\ta)}^{1} \frac{\partial \eta (\theta)}{\partial \pu}  \, dμ + \int_0^{\eta ^{-1}(\ta)} \frac{\partial  \ta}{\partial \pu} \, dμ= 0
% \end{aligned}
% \] 
Since $\frac{\partial \eta (\theta)}{\partial \pu} > 0$ for all $\theta \in [0,1]$, 
we obtain $\frac{\partial \ta}{\partial \pu} < 0$.

For $\pmap^*(\theta)$, recall that
$\pmap^*(\theta) = 1 - \frac{\ta- \pv}{\eta (\theta) - \pv}$. Taking
derivative with respect to $\pu$ gives \[\begin{aligned}
        \frac{\partial \pmap^*(\theta)}{\partial \pu}
         = - \left( \frac{ \frac{\partial \ta}{\partial \pu} (\eta (\theta) - \pv) - (\ta- \pv) \frac{\partial \eta(\theta)}{\partial \pu}}{(\eta (\theta) - \pv)^2} \right).
    \end{aligned}\] Since $\frac{\partial \ta}{\partial \pu} < 0$,
$\eta (\theta) - \pv> 0$, $\ta- \pv> 0$ and
$\frac{\partial \eta(\theta)}{\partial \pu} > 0$, we have
$\frac{\partial \pmap^*(\theta)}{\partial \pu} > 0$.

\subsection{Proof of Proposition~\ref{prp-cs-v}} \label{app-a4}

Recall $\eta (\theta) = \sqrt{\pu+ (\theta - \pv)^2} + \pv$. Partial
derivative of $\eta (\theta)$ with respect to $\pv$ yields:
\[\frac{\partial \eta (\theta)}{\partial \pv} = - \left( \frac{\theta - \pv}{\sqrt{(\theta - \pv)^2} + \pu} - 1 \right) < 0 \text{ for all } \theta \in [0,1].\]
% Therefore, $\eta (\theta)$ decreases with $\pv$ for all
% $\theta \in [0,1]$.

As for $\ta$, recall that it is obtained by the FOC. Letting $\Gamma_2'$ equal to $0$ yields:
\[\int_{\eta ^{-1}(\ta)}^{1}  (\eta (\theta) - \ta) \, dμ  + \ta- \mathbb{E} (\theta) = 0.\]
Taking partial derivative with respect to $\pv$ yields:
$$
\int_{\eta ^{-1}(\ta)}^{1} \frac{\partial \eta (\theta)}{\partial \pv}  \, dμ + \int_0^{\eta ^{-1}(\ta)} \frac{\partial  \ta}{\partial \pv} \, dμ = 0.
$$
% \[
% \begin{aligned}
%         & \frac{\partial \eta ^{-1}({\ta})}{\partial \pv} (\eta (\hat{\theta}) - \ta) g(\hat{\theta}) + \int_{\eta ^{-1}(\ta)}^{1} \frac{\partial \eta (\theta)}{\partial \pv} - \frac{\partial  \ta}{\partial \pv} \, dμ + \frac{\partial  \ta}{\partial \pv} = 0\\
%         \implies & \int_{\eta ^{-1}(\ta)}^{1} \frac{\partial \eta (\theta)}{\partial \pv}  \, dμ + \int_0^{\eta ^{-1}(\ta)} \frac{\partial  \ta}{\partial \pv} \, dμ= 0
% \end{aligned}
% \]
As we have established that
$\frac{\partial \eta (\theta)}{\partial \pv} < 0$ for all
$\theta \in [0,1]$, $\frac{\partial \ta}{\partial \pv} > 0$ follows.

\subsection{Veto probability $\pmap^* (θ)$ may not change uniformly as $\vn$ varies}\label{app-a4-numeric}

We use two numeric examples to illustrate
that $p^*(\theta)$ may or may not change uniformly when $\vn$ varies.
\begin{enumerate}
\item 
Suppose the density function is $g(\theta) = 1$ for $\theta \in \Theta$ and $\pu= -0.9$.
Consider the two cases:
$\vn = -0.9$ and $\vn' = -0.3$.
It follows from \Cref{prp-char} that
$p^*(\theta)$ weakly decreases  
for all $\theta \in \Theta$ as  
$\vn$ increases to $\vn'$.
The left panel of \Cref{fig-not-uniform-p} illustrates this,
where the dashed and solid curves represent the optimal veto probability under $\vn$ and $\vn'$ respectively.

\item 
Suppose the density function $g(\theta) = 1$ for $\theta \in \Theta$ and $\pu= -0.2$.
Consider the two cases:
$\vn = -0.5$ and $\vn' = -0.3$.
When the status quo option increases from $\vn$ to $\vn'$, the optimal veto probability at state $θ=1$ 
increases from around $0.42$ to around $0.47$.
Meanwhile, the cutoff also increases from 
$\ttheta \approx 0.44$ to $\ttheta' \approx 0.49$.
The right panel of \Cref{fig-not-uniform-p} illustrates this,
where the dashed and solid curves represent the optimal veto probability under $\vn$ and $\vn'$ respectively.
%It follows from \Cref{prp-char} that as agent's payoff from the status quo option increases from $\vn$ to $\vn'$.

\end{enumerate}

% The change in $p^*(\theta)$
% can be either positive or negative, depending on the state.

\begin{figure}
\begin{center}

  \begin{tikzpicture}[domain=0.001:1, scale=6.6, xscale=1,
    % every node/.style={scale = 1}
   ] 
 %% vertices
 \draw[->, ] (0,0) node[left]{$0$} -- (0.66,0) node[right] {$\theta$}; 
 \draw[->, ] (0,0) -- (0,0.66) node[above] {$\tilde{p}(\theta)$};
 %% dotted lines
 \draw[dotted, ] (0.6,0) -- (0.6, 0.6);
 \draw[dotted, ] (0,0.6) node [left] {$1$}-- (0.6,0.6);
 %% 1
 \draw [ ] (0.6, 0) node [below] {1}-- (0.6, 0.02);
 %% p(\theta, v^*) v = -0.6, u = -0.2
 \draw[, domain=0.33:0.6]   plot (\x,{1 - (0.215+0.6)/(sqrt((\x + 0.6) *(\x + 0.6) -0.2 ))}); 
 %% p(\theta, v^*) v = -0.6, u = -0.1
 \draw[dashed,very thick  , domain=0.22:0.6]   plot (\x,{1 - (0.1566+0.6)/(sqrt((\x + 0.6) *(\x + 0.6) -0.1 ))});

 \draw[]  (0, 0) -- (0.33, 0) node [below, ]{$\color{black}\bar \theta'$};
 \draw[dashed, very thick, ]  (0, 0) -- (0.22, 0) node [below, ]{$\color{black}\bar \theta $};
 %% arrows
 \end{tikzpicture}
%%%%%%%%% 
\qquad
%%%%%%%%% 
\begin{tikzpicture}[domain=0.001:1, scale=4, xscale=1] 
      %% vertices
        \draw[-> ] (0,0) node[left]{$0$} -- (1.1,0) node[right] {$\theta$}; 
      \draw[-> ] (0, 0) -- (0,1.1) node[above] {$\pmap (\theta)$};
      %% v^*
 %     \draw[color=red, thick]  (0, 0.2) node [left] {$\color{red}v^*$} -- (0.02, 0.2); 
      %% dotted lines
        \draw[dotted ] (1,0) -- (1,1);
      \draw[dotted ] (0,1) node [left] {$1$}-- (1,1);
      %% 1
      \draw [thick] (1, 0) node [below] {1}-- (1, 0.02);
      %% p(\theta, v^*) v = -0.3, u = -0.2
      \draw[dashed, very thick, domain=0.42:1]   plot (\x,{1 - (0.264+0.3)/(sqrt((\x + 0.3) *(\x + 0.3) -0.2 ))}); 
      \draw[dashed, very thick]  (0, 0) -- (0.42, 0) node [below, xshift=0.1cm]{$\color{black} \bar \theta'$};
      %% p(\theta, v^*) v = -0.6, u = -0.2
      \draw[color=black,  domain=0.33:1]   plot (\x,{1 - (0.215+0.6)/(sqrt((\x + 0.6) *(\x + 0.6) -0.2 ))}); 
      \draw[color=black]  (0, 0) -- (0.33, 0) node [below]{$\color{black}\bar \theta$};
      %% arrows
      %\draw [->,   brown] (0.3, -0.18) -- (0.42, -0.18);
      %\draw [->,   brown] (1.32, 0.45) -- (1.32, 0.58);
         \draw (1.02, 0.46) node [right] {};
      \draw (1.02, 0.58) node [right] {};
\end{tikzpicture}  
\end{center}
\caption{Changes of $\pmap^*(θ)$ when $\pv$ varies}
\label{fig-not-uniform-p}  
\end{figure}

\subsection{Insufficiency of the likelihood ratio conditions}\label{app-a6}

We use numeric examples to illustrate that a higher likelihood ratio cannot guarantee either a higher or a lower $\ta$. First, consider two distributions on $\Theta = [0,1]$ with density functions
$g_L$ and $g_H$ as follows: 
\[
g_L (\theta) = \begin{cases}
  1/3           & \text{ if } \theta \in [0,0.9];    \\
  140(1-\theta) & \text{ if } \theta \in (0.9,1].  \\
\end{cases}
\]
\[
g_H (\theta) = \begin{cases}
  1/3           & \text{ if } \theta \in [0,0.9];    \\
  140(\theta-0.9) & \text{ if } \theta \in (0.9,1].  \\
\end{cases}
\]
It follows that
$g_H$ dominates $g_L$ in the sense of likelihood ratio. Fixing
$\pu= -0.2$ and $\pv= -0.6$, the corresponding threshold values are
$\bar\theta_H \approx 0.647$ and $\bar\theta_L \approx 0.652$.
So $\ttheta_H < \ttheta_L$.

On the other hand, consider two density functions 
$\hat{g}_L$ and $\hat{g}_H$ as follows:
\[
\hat{g}_L(\theta) = 1 \text{ and }
\hat{g}_H(\theta) = 10 \theta^9, \forall \theta \in [0,1].
\]
It follows that $\hat{g}_H$ dominates 
$\hat{g}_L$ in
the sense of likelihood ratio. 
Fixing $\pu= -0.2$ and $\pv= -0.6$,
the corresponding threshold values are 
$\ttheta_H' \approx 0.970$ and
$\ttheta_L' \approx 0.424$. So $\ttheta'_H > \ttheta'_L$.

\subsection{Proof of \Cref{prp-ex-optimal-mechanism}\label{app-a7}}

The proof of \Cref{prp-ex-optimal-mechanism} is similar to that of \Cref{prp-char}.

\subsubsection*{Step 1: Pointwise
Optimization.}

Fixing $\av \leq \pv$, constraints~\eqref{eq-ic} imply
\begin{equation}\protect\hypertarget{eq-eg-prob-function2}{}{
  \pmap(\theta) =  \frac{\av - \amap(\theta)}{\pv- \amap(\theta)}.
}\label{eq-eg-prob-function2}\end{equation}
Equation~\eqref{eq-eg-prob-function2} implies that $\pmap(\theta) \le 1$ is
equivalent to $\av \le \underline v$ and that $\pmap(\theta) \ge 0$ is
equivalent to $\av \ge \amap(\theta)$. 
Substituting Equation~\eqref{eq-eg-prob-function2} into the
principal's utility, the optimization problem \eqref{model:dm3} is reduced to
%\begin{equation}\protect\hypertarget{eq-eg-optimization2}{}{
%\begin{split}
%  &\max_{\amap(\theta) \in A}\,  \Big(\frac{\av - \amap(\theta)}{\pv- \amap(\theta)} \Big) \underline u - 
%   \Big(\frac{\pv- \av}{\pv- \amap(\theta)}\Big) (\amap(\theta) -\theta)^2    \\
%  &\text{ subject to }  \amap(\theta) \le \av, \forall \theta \in \Theta. 
%\end{split}
%}\label{eq-eg-optimization2}\end{equation} 
\begin{equation}
\begin{split}
  &\max_{\amap(\theta) \in A}\,  \Big(\frac{\av - \amap(\theta)}{\pv- \amap(\theta)} \Big) \underline u - 
   \Big(\frac{\pv- \av}{\pv- \amap(\theta)}\Big) (\amap(\theta) -\theta)^2    \\
  &\text{ subject to }  \amap(\theta) \le \av, \forall \theta \in \Theta. 
\end{split}\label{eq-eg-optimization2}    
\end{equation}
The solution to
problem~\eqref{eq-eg-optimization2} is
\begin{equation}\protect\hypertarget{eq-optimal-a2}{}{
\amap^*(\theta) = \begin{cases}
    \av & \text{ if } (\pv- \theta)^2 + \pu< 0,  \\
    \min \{ \av, \varphi (\theta) \} & \text{ otherwise } 
\end{cases}
}\label{eq-optimal-a2}\end{equation} 
where
$\varphi (\theta) = \pv - \sqrt{( \pv- \theta)^2 + \pu}$.

\hypertarget{step-2-optimal-v.-1}{%
\subsubsection*{\texorpdfstring{Step 2: Optimal
$\av$.}{Step 2: Optimal.}}\label{step-2-optimal-v.-1}}

Based on Equation~\eqref{eq-optimal-a2}, we derive the optimal $\av$ for different values of $\pu$ and $\pv$ characterized by the following three cases:

\begin{enumerate}[(i)]
\item 
  $ \pu+ (\pv - 1) ^ 2 \ge 0$. In this case, $\pu + (\pv - \theta) ^ 2 \geq 0$ for all $\theta \in [0,1]$.
\item 
  $\pu + \pv ^ 2 > 0$ and $ \pu + ( \pv - 1) ^2 < 0$. In this case, $\pu + (\pv - \theta) ^2 \geq 0$ for $\theta \in [0, \sqrt{- \pu} + \pv]$ and
  $ \pu + (\theta - \pv) ^2< 0$ for
  $\theta \in (\sqrt{- \pu} + \pv , 1]$.
\item 
  $ \pu + \pv ^2 \le 0$. In this case, $\pu +(\pv - \theta)^2 \le 0$ for all $\theta \in [0,1]$.
\end{enumerate}

The solutions of case (i) and 
case (ii) are the same:
\begin{enumerate}
\def\labelenumi{\arabic{enumi}.}
\item
  If $\varphi(0) \geq \mathbb{E}_μ (\theta)$, $\av^*= \mathbb{E}_μ (\theta)$.  The actions are pooled at $\amap(\theta) = \mathbb{E}_μ (\theta)$.
\item
  If $\varphi(0) < \mathbb{E}_μ (\theta)$, $\av^* = \bar{v}$ where $\bar v$ is given by 
  \[
  \int_0^{\varphi ^{-1}(\bar v)} \varphi (\theta) - \bar v \, dμ + \bar v - \mathbb{E}_μ (\theta) = 0
  \]
  The action function $\amap(\theta)$ is as follows: 
  \begin{itemize}
  \item 
    when $\theta \in [0, \varphi ^{-1}(\bar v))$, the default action $a_0$ is chosen with probability $\frac{\bar v - \varphi(\theta)}{\pv- \varphi(\theta)}$ and the action $\amap(\theta) = \varphi(\theta)$ is chosen with the complementary probability;
  \item 
    when $\theta \in [\varphi ^{-1}(\bar v), 1]$,
    $\amap(\theta) = \bar v$.
  \end{itemize}
  {We could further show that $\bar v \in (\mathbb{E}_μ (\theta), 1)$.
  Let 
  $f(\hat v) =  \int_0^{\varphi ^{-1}(\hat v)} \varphi (\theta) - \hat v \, dμ + \hat v - \mathbb{E}_μ (\theta).$
  Then, $\bar v$ is the solution to $f(\hat v) = 0$.
  $f(\hat v)$ is strictly increasing:
  \[
  f'(\hat v) = \int_{\varphi ^{-1}(\hat v)}^1 1 \, d μ (\theta) > 0.
  \]
  At the two points $\hat v = \mathbb{E}_μ (\theta)$ and $\hat v = 1 $: 
  \[ \begin{split}
       f(\mathbb{E}_μ (\theta)) & = -  \int_0^{\varphi ^{-1}(\mathbb{E}_μ (\theta))}   \mathbb{E}_μ (\theta) - \varphi (\theta) \, dμ < 0 \\
      f(1) & =   \int_0^{\varphi ^{-1}(1)} \varphi (\theta) - 1 \, dμ + 1 - \mathbb{E}_μ (\theta) \\
      & = \int_0^{\varphi ^{-1}(1)} \varphi (\theta) - \theta \, dμ + \int_{\varphi ^{-1}(1)}^1 1 - \theta \, dμ >0
  \end{split} \]
  It follows that $\bar v \in (\mathbb{E}_μ (\theta), 1)$.
  }
\end{enumerate}
The condition for cases (i) and (ii) combined is $\pu + \pv^2 > 0$.
Furthermore, the condition for the pooled action $\varphi(0) \geq \mathbb{E}_μ (\theta)$ is equivalent to $( \pv- \mathbb{E}_μ (\theta))^2 \geq \pu + \pv^2$.
Therefore, 
\begin{itemize}
\item 
  When $ ( \pv- \mathbb{E}_μ (\theta))^2 \geq \pu + \pv^2 > 0$, the actions are pooled at $\amap^*(\theta) = \mathbb{E}_μ (\theta)$;
\item 
  When $(\pv- \mathbb{E}_μ (\theta))^2 < \pu + \pv^2 $, the optimal action rule $(\amap^*, \pmap^*)$ is as characterized in part~(a) of \Cref{prp-ex-optimal-mechanism}.
\end{itemize}

For case (iii), the principal optimally sets $\av^* = \mathbb{E}_μ (\theta)$, resulting in pooled action $\amap^*(\theta) = \mathbb{E}_μ (\theta)$.

\Cref{prp-ex-optimal-mechanism} summarizes the analysis above:
the cases $( \pv- \mathbb{E}_μ (\theta))^2 \geq \pu + \pv^2 > 0$ and $\pu + \pv ^2 \le 0$ both result in the pooled action and are stated in part~(b) of the proposition;
the case $(\pv- \mathbb{E}_μ (\theta))^2 < \pu + \pv^2$ corresponds to part~(a) of the proposition.
% Proposition \ref{prp-ex-optimal-mechanism} summarizes the analysis above. 
% In particular, the cases  $ ( \pv- \mathbb{E}_μ (\theta))^2 \geq \pu + \pv^2 > 0$ and $ \pu + \pv ^2 \le 0$ both result in the pooled action and are stated in part~(b) of the proposition.
% The case $ ( \pv- \mathbb{E}_μ (\theta))^2 < \pu + \pv^2$ corresponds to part~(a) of the proposition.
% For expositional purpose, we use $\bar a$ to denote $\bar{v}$ in the proposition.

\subsection{Derivations omitted in \Cref{sec:extension-discuss}}
\hypertarget{app-a81}{%
\subsubsection{U-shaped $\amap^*(\theta)$ and $\pmap^*(\theta)$}\label{app-a81}}
The parameters under concern are $\pv = 0.38$, $\pu = -0.1$ and $g(\theta) = 1$ for all $\theta \in  [0,1]$.
\Cref{prp:main} still applies. 
We can focus on the class of veto mechanisms when searching for a principal's optimal mechanism
and decompose principal's maximization problem into two sequential problems \eqref{model:dm3} and \eqref{model:dm4}.
The solution to \eqref{model:dm3} is different when $\av \geq \pv$ and $\av \leq \pv$. And we consider the two cases separately.

\paragraph{Case I: $\av \geq \pv$.}
When $\av \geq \pv$, we have $\amap(\theta) \geq \pv$ for all $\theta \in [0,1]$. 
The solution to \eqref{model:dm3} is the same as the solution for the case $\pv \le 0$:
\[
  \amap^*(\theta) = \begin{cases}
    \av & \text{ if }  (\theta - \pv)^2  + \pu < 0,  \\
    \max \set{\av, \eta(\theta)} & \text{ otherwise } 
  \end{cases}
\]
where $\eta(\theta) = \sqrt{(\theta - \pv)^2 + \pu} + \pv$.
Plugging in the parameter values, we have
\begin{enumerate}
  \item If $\av \in [\vn, \sqrt{0.0444} + 0.38]$,
\[
  \amap^*(\theta) = \begin{cases}
    \av & \text{ if } \underline{\theta} < \theta <  \bar{\theta},  \\
     \sqrt{\theta - 0.38)^2 -0.1} + 0.38 & \text{ otherwise } 
  \end{cases}
\]
where $\underline{\theta} =0.38-\sqrt{(\av - 0.38)^2 + 0.1} $ and $\bar{\theta} = 0.38+\sqrt{(\av - 0.38)^2 + 0.1} $.
  \item If $\av \in [\sqrt{0.0444} + 0.38, \sqrt{0.2844} + 0.38]$,
\[
  \amap^*(\theta) = \begin{cases}
    \av & \text{ if } \theta <  \bar{\theta},  \\
     \sqrt{(\theta - 0.38)^2 -0.1} + 0.38 & \text{ otherwise } 
  \end{cases}
\]
where $\bar{\theta} = 0.38+\sqrt{(\av - 0.38)^2 + 0.1}$.
  \item If $\av \geq \sqrt{0.2844} + 0.38$,
  $\amap^*(\theta) = \av$.
\end{enumerate}
Next, we solve the maximization problem \eqref{model:dm4}.
The optimal $\av$ is within the range $\av \in [\vn, \sqrt{0.0444} + 0.38]$ and given by the below FOC:
\begin{align} \label{eq-a-ushape}
  \int_0^{\underline{\theta}} \eta(\theta) - \av \, d\theta + \int_{\bar{\theta}}^1 \eta(\theta) - \av \, d\theta + \av - \frac{1}{2} = 0
\end{align}
where $\eta (\theta) = \sqrt{(\theta - 0.38)^2 -0.1} + 0.38 $, $\underline{\theta} =0.38-\sqrt{(\av - 0.38)^2 + 0.1} $ and $\bar{\theta} = 0.38+\sqrt{(\av - 0.38)^2 + 0.1}$
as calculated when solving \eqref{model:dm3}.
Solving \Cref{eq-a-ushape} numerically gives
$\av^* \approx 0.397, \underline{\theta} \approx 0.063$ and 
$\bar{\theta} \approx 0.697$.

\paragraph{Case II: $\av \leq \pv$.}
When $\av \leq \pv$, we have $\amap(\theta) \leq \pv$ for all $\theta \in [0,1]$. 
The solution to \eqref{model:dm3} is the same as the solution for the case $\pv \ge 1$.

\[
\amap^*(\theta) = \begin{cases}
    \av & \text{ if } (\pv- \theta)^2 + \pu< 0,  \\
    \min \{ \av, \varphi (\theta) \} & \text{ otherwise } 
\end{cases}
\]
 where
$\varphi (\theta) = \pv - \sqrt{( \pv- \theta)^2 + \pu}$.
Plugging in the parameter values, we have
\begin{enumerate}
  \item If $\av \in [0.38 -\sqrt{0.0444} ,  \vn]$,
\[
  \amap^*(\theta) = \begin{cases}
    \av & \text{ if } \underline{\theta} < \theta <  \bar{\theta},  \\
     0.38 - \sqrt{(\theta - 0.38)^2 -0.1}  & \text{ otherwise } 
  \end{cases}
\]
where $\underline{\theta} =0.38-\sqrt{(\av - 0.38)^2 + 0.1} $ and $\bar{\theta} = 0.38+\sqrt{(\av - 0.38)^2 + 0.1} $.
  \item If $\av \in [0.38 -\sqrt{0.2844} , 0.38 -\sqrt{0.0444} ]$,
\[
  \amap^*(\theta) = \begin{cases}
    \av & \text{ if } \theta <  \bar{\theta},  \\
     0.38 - \sqrt{(\theta - 0.38)^2 -0.1}  & \text{ otherwise } 
  \end{cases}
\]
where $\bar{\theta} = 0.38+\sqrt{(\av - 0.38)^2 + 0.1}$.
  \item If $\av \leq 0.38 -\sqrt{0.2844} $,
  $\amap^*(\theta) = \av$.
\end{enumerate}
Next, we solve the maximization problem \eqref{model:dm4}.
Note that the optimal $\av$ is within the range $\av \in [0.38 -\sqrt{0.0444} ,  \vn]$ and that the first-order derivative 
is positive for all $\av \in [0.38 -\sqrt{0.0444} , \pv = 0.38]$:
\[
\begin{split} 
  2 \big[ \int_0^{\underline{\theta}}  \av - \varphi(\theta) \, d\theta + \int_{\bar{\theta}}^1 \av - \varphi(\theta) \, d\theta + ( \frac{1}{2} -\av) \big] > 0.
\end{split} 
\]
Therefore, we obtain the corner solution $\av^* = \pv  = 0.38$.

Combining the two cases, the optimal solution is obtained when $\av^* \approx 0.397$.

\hypertarget{app-a82}{%
\subsubsection{Deterministic optimal mechanism}\label{app-a82}}

%The parameters under concern are $\pv = \frac{1}{2}$, $\pu > -\frac{1}{4}$ and $g(\theta) = 1$ for all $\theta \in [0,1]$.

We follow the same procedures as in the previous proof of 
\Cref{app-a81}.

\paragraph{Case I: $\av \geq \pv$.} 
The solution to \eqref{model:dm3} is the same as the solution for the case $\pv \le 0$.
%Given $\pv = \frac{1}{2}$, $\pu > -\frac{1}{4}$ and $g(\theta) = 1$, we have
\begin{enumerate}
  \item If $\av \in [\pv = 0.5, \sqrt{0.25 +\pu} + 0.5]$,
\[
  \amap^*(\theta) = \begin{cases}
    \av & \text{ if } \underline{\theta} \leq \theta \leq  \bar{\theta},  \\
     \sqrt{(\theta - 0.5)^2 +\pu} + 0.5 & \text{ otherwise } 
  \end{cases}
\]
where $\underline{\theta} = 0.5-\sqrt{(\av - 0.5)^2 -\pu} $ and $\bar{\theta} = 0.5 + \sqrt{(\av - 0.5)^2 -\pu} $.
  \item If $\av \geq \sqrt{0.25 +\pu} + 0.5$,
  $\amap^*(\theta) = \av$.
\end{enumerate}
Next, we solve the maximization problem \eqref{model:dm4}.
The optimal $\av$ is within the range $[\pv, \sqrt{0.25 +\pu} + 0.5]$ and the first-order derivative is negative
for all $\av \in [\pv, \sqrt{0.25 +\pu} + 0.5]$:
\[ \begin{split} 
  - 2 \big[ \int_0^{\underline{\theta}}  \eta(\theta) - \av \, d\theta + \int_{\bar{\theta}}^1 \eta(\theta) - \av  \, d\theta + (\av - \frac{1}{2} ) \big] < 0.
\end{split} \]
Therefore, we obtain a corner solution $\av^* = \pv  = 0.5$.

\paragraph{Case II: $\av \leq \pv$.}
The solution to \eqref{model:dm3} is the same as the solution for the case $\pv \ge 1$.
%Given the parameter values, we have
\begin{enumerate}
  \item If $\av \in [0.5 - \sqrt{0.25 +\pu}, \pv  = 0.5]$,
\[
  \amap^*(\theta) = \begin{cases}
    \av & \text{ if } \underline{\theta} \leq \theta \leq  \bar{\theta},  \\
     0.5 - \sqrt{(\theta - 0.5)^2 + \pu}  & \text{ otherwise } 
  \end{cases}
\]
where $\underline{\theta} =0.5-\sqrt{(\av - 0.5)^2 -\pu} $ and $\bar{\theta} = 0.5 + \sqrt{(\av - 0.5)^2 - \pu} $.
  \item If $\av \leq 0.5 - \sqrt{0.25 +\pu}$,
  $\amap^*(\theta) = \av$.
\end{enumerate}
Next, we solve the maximization problem \eqref{model:dm4}.
The optimal $\av$ is within the range $\av \in [0.5 - \sqrt{0.25 +\pu}, \pv]$ and 
the first-order derivative is positive for all $\av \in [0.5 - \sqrt{0.25 +\pu}, \pv]$:
\[ \begin{split} 
  2 \big[ \int_0^{\underline{\theta}}  \av - \varphi(\theta) \, d\theta + \int_{\bar{\theta}}^1 \av - \varphi(\theta) \, d\theta + ( \frac{1}{2} -\av) \big] > 0.
\end{split} \]
Therefore, we obtain a corner solution $\av^* = \pv$.

Combining the two cases, we obtain that in the solution $\av^* = \pv$.
Then $\underline{\theta}^* = 0.5 - \sqrt{-\pu}$ and $\bar{\theta}^* = 0.5 + \sqrt{-\pu}$.
And both cases correspond to the veto probability:
\[
  \pmap^*(\theta) = \begin{cases}
    0 & \text{ if } \underline{\theta}^* \leq \theta \leq  \bar{\theta}^*,  \\
    1 & \text{ otherwise. } 
  \end{cases}
\]
Since the principal vetoes with probability $1$ when $\theta \in [0, \underline{\theta}^*)$ and $ \theta \in (\bar{\theta}^*, 1]$,
$\amap^*(\theta)$ can take any value.
When $\theta \in [\underline{\theta}^*, \bar{\theta}^*]$, we have $\amap^*(\theta) = \av^* = \vn$.

\end{document}